\documentclass[11pt,letterpaper]{article}
\usepackage[margin=1in]{geometry}
\usepackage{amsmath, amsfonts, amssymb, amsthm}

\usepackage{cite}

\usepackage[dvipsnames]{xcolor}
\usepackage{bbm}
\usepackage{tikz} 
\usepackage{graphicx}
\usepackage{authblk}
\usepackage{caption}
\usepackage{subcaption}

\usepackage{complexity}

\usepackage{epstopdf}
\usepackage{wrapfig}
\usepackage{hyperref}

\usepackage{xcolor}
\usepackage{tikz} 

\usepackage{graphicx}
\usepackage{caption}
\usepackage{subcaption}
\usepackage{mathtools}
\usepackage{mathrsfs}

\usepackage{listings}
\usepackage[framed, numbered]{matlab-prettifier}

\usepackage{hyperref}
\usepackage[shortlabels]{enumitem}

\usepackage{algorithm}
\usepackage{algpseudocode}

\usepackage{booktabs}
\usetikzlibrary{arrows.meta,calc,positioning,fit,backgrounds}

\newcommand{\zo}{\{0,1\}}

\DeclareMathOperator{\Rec}{\mathsf{Rec}}

\DeclareMathOperator{\Supp}{Supp}

\newcommand{\abs}[1]{\left| #1 \right|} 

\theoremstyle{definition}

\newtheorem*{definition}{Definition}

\theoremstyle{plain}  
\newtheorem{theorem}{Theorem}
\newtheorem{lemma}[theorem]{Lemma}

\newtheorem{corollary}[theorem]{Corollary}

\theoremstyle{remark}  
\newtheorem{rem}[theorem]{Remark}
\newtheorem{remark}[theorem]{Remark}

\def\ddefloop#1{\ifx\ddefloop#1\else\ddef{#1}\expandafter\ddefloop\fi}

\def\ddef#1{\expandafter\def\csname bb#1\endcsname{\ensuremath{\mathbb{#1}}}}
\ddefloop ABCDEFGHIJKLMNOPQRSTUVWXYZ\ddefloop

\def\ddef#1{\expandafter\def\csname c#1\endcsname{\ensuremath{\mathcal{#1}}}}
\ddefloop ABCDEFGHIJKLMNOPQRSTUVWXYZ\ddefloop

\def\ddef#1{\expandafter\def\csname tt#1\endcsname{\ensuremath{\mathtt{#1}}}}
\ddefloop ABCDEFGHIJKLMNOPQRSTUVWXYZ\ddefloop

\def\ddef#1{\expandafter\def\csname scr#1\endcsname{\ensuremath{\mathscr{#1}}}}
\ddefloop ABCDEFGHIJKLMNOPQRSTUVWXYZ\ddefloop

\def\ddef#1{\expandafter\def\csname v#1\endcsname{\ensuremath{\boldsymbol{#1}}}}
\ddefloop ABCDEFGHIJKLMNOPQRSTUVWXYZabcdefghijklmnopqrstuvwxyz\ddefloop

\def\ddef#1{\expandafter\def\csname v#1\endcsname{\ensuremath{\boldsymbol{\csname #1\endcsname}}}}
\ddefloop {alpha}{beta}{gamma}{delta}{epsilon}{varepsilon}{zeta}{eta}{theta}{vartheta}{iota}{kappa}{lambda}{mu}{nu}{xi}{pi}{varpi}{rho}{varrho}{sigma}{varsigma}{tau}{upsilon}{phi}{varphi}{chi}{psi}{omega}{Gamma}{Delta}{Theta}{Lambda}{Xi}{Pi}{Sigma}{varSigma}{Upsilon}{Phi}{Psi}{Omega}{ell}\ddefloop

\newcommand{\F}{{\mathbb{F}}}

\newcommand{\N}{{\mathbb{N}}}

\definecolor{gray85}{gray}{0.85} 
\definecolor{gray8}{gray}{0.8} 
\definecolor{gray7}{gray}{0.7} 
\definecolor{gray6}{gray}{0.6} 
\definecolor{gray5}{gray}{0.5} 
\definecolor{gray4}{gray}{0.4} 
\definecolor{gray35}{gray}{0.35} 

\definecolor{forestgreen}{rgb}{0.13, 0.55, 0.13}

\definecolor{darkgray}{rgb}{.55,.55,.55}
\definecolor{lightgray}{rgb}{.8,.8,.8}
\definecolor{verylightgray}{rgb}{1,1,1} 

\graphicspath{{Figures/}}

\newcommand{\vladimir}[1]{\textcolor{blue}{\bf !}\marginpar{\parbox{20mm}{\scriptsize\raggedright\textcolor{blue}{Vladimir: #1}}}}
\newcommand{\marshall}[1]{\textcolor{red}{\bf !}\marginpar{\parbox{20mm}{\scriptsize\raggedright\textcolor{red}{Marshall: #1}}}}
\newcommand{\morgan}[1]{\textcolor{violet}{\bf !}\marginpar{\parbox{20mm}{\scriptsize\raggedright\textcolor{violet}{Morgan: #1}}}}

\renewcommand{\vladimir}[1]{}
\renewcommand{\marshall}[1]{}
\renewcommand{\morgan}[1]{}

\newcommand{\OR}{\mathsf{OR}}
\newcommand{\NOR}{\mathsf{NOR}}
\newcommand{\AND}{\mathsf{AND}}
\newcommand{\MAJ}{\mathsf{MAJ}}
\newcommand{\XOR}{\mathsf{XOR}}
\newcommand{\EMAJ}{\mathsf{EMAJ}}
\newcommand{\Cube}{\mathtt{Cube}}
\newcommand{\Image}{\mathrm{Im}}

\newcommand{\Cindiv}{C_{\text{indiv}}}
\newcommand{\Csum}{C_{\text{indiv-sum}}}

\newcommand{\ZZ}{\mathbb{Z}}

\newcommand{\CRT}{\mathsf{CRT}}

\title{On the Complexity of Symmetric Boolean Functions \\in Multiparty Private Simultaneous Messages}
        \author[1]{Marshall Ball}
	\author[1]{Peter Crawford-Kahrl}
        \author[2]{Vladimir Podolskii}
        \author[2]{Morgan Prior}
        \affil[1]{New York University (\texttt{marshall@cs.nyu.edu}, \texttt{peter.crawford.kahrl@gmail.com})}
        \affil[2]{Tufts University (\texttt{podolskii.vv@gmail.com}, \texttt{morgan.prior@tufts.edu})}
\date{}

\begin{document}
	
	\maketitle
    \begin{abstract}
    \normalsize
        Multiparty Private Simultaneous Messages (PSM), also known as Decomposable Randomized Encoding (DRE) or information-theoretic garbling, is a model for non-interactive secure computation. Each party encodes its input using shared randomness and sends a message to an evaluator who does not know that randomness. The evaluator must recover the function output without learning anything else about the inputs. Perfect privacy requires identical transcript distributions on inputs with the same output. We also study the weaker notion of \emph{support privacy}, which requires only that these distributions have the same support.

        We study the communication complexity of symmetric Boolean functions when each of the $n$ parties holds one input bit.

        {\it Lower bounds.} We prove a $(1-o(1))n\log n$ lower bound for support-private PSM protocols computing $\OR_n$, and as a consequence for $\MAJ_n$ and $\EMAJ_n$ as well. We also prove an $\Omega(n\log q)$ lower bound for $\MOD_q$ when $2\le q\le n$, with a constant independent of $q$. An alternative proof technique recovers most of these bounds as well as the previously known approach based on Nechiporuk's method.

        {\it Upper bounds.} We construct perfectly private PSM protocols of total communication $\widetilde{O}(n^{3/2})$ for every symmetric Boolean function, improving on the known quadratic bound. For $\MAJ_n$, known branching-program constructions yield an upper bound of $O(n\log^3 n\log\log n)$. For $\EMAJ_n$ and $\MOD_q$ with prime $q$, the lower bounds are tight up to constant factors. The matching upper bound for $\OR_n$ was previously known. Under support privacy, we obtain $\widetilde{O}(n)$ communication for every symmetric Boolean function.

        Earlier lower-bound techniques from~\cite{ball2020complexity} yield only trivial linear bounds for symmetric Boolean functions. Concurrent and independent work by Holmgren and Wu proves a $\omega(n)$ lower bound for $\OR_n$ using the same approach as our alternative technique. Both that work and the present work used substantial AI assistance.

    \end{abstract}

    \newpage

\section{Introduction}
    Feige, Kilian, and Naor~\cite{FKN} introduced a minimal model for secure computation, later named Private Simultaneous Messages (PSM) by Ishai and Kushilevitz~\cite{ishai1997private}. The model is closely related to Decomposable Randomized Encodings (DRE)~\cite{ishai2002perfect,applebaum2006cryptography,ishai2013randomization} and information-theoretic garbling~\cite{yao1986generate}. We refer to~\cite{bellare2012foundations,ishai2013randomization,applebaum2017garbled} for surveys of these objects and their cryptographic applications.

    In a PSM protocol for a function $f$, each of $n$ parties holds an input $x_i$ and locally computes a message $m_i$ using a shared random string $r$. The parties send their messages to an evaluator, who has access to neither the inputs nor the shared randomness and must compute $f(x_1,\ldots,x_n)$. Privacy requires that the evaluator learn nothing else about the inputs. We study two information-theoretic notions of privacy, both defined through the \emph{transcript}: the tuple of messages sent by all parties. Perfect privacy requires identical transcript distributions whenever $f(x)=f(x')$. Support privacy requires only that these distributions have the same support.

    A central goal in this setting is to understand \emph{the cost of privacy}. We study this cost through communication complexity. For a Boolean function $f$, we distinguish the individual message complexity $C_{\text{indiv}}(f)$ from the total communication complexity $C(f)$. The former measures the logarithm of the largest local message alphabet, minimized over protocols for $f$. The latter measures the length of the longest transcript, again minimized over protocols. We give the precise definitions in Section~\ref{subsec:metrics-lifting}.
    
    Without privacy, each party can send its input directly to the evaluator. When each $x_i\in\{0,1\}$, this costs $n$ bits in total. Only a few functions, such as $\XOR_n$, are known to admit private protocols with just $n$ bits of communication.
    
    For the general setting, the paper~\cite{BeimelKN18} showed that any function on $n$ bits can be computed with communication complexity $O(n^3 2^{n/2})$. This was later slightly improved by Assouline and Liu~\cite{assouline2021multi} when the input of each player is large. On the other hand, the authors in~\cite{ball2020complexity} showed that almost all functions require communication complexity $\Omega(n^2/\log n)$, and exhibited explicit functions with such cost (e.g.~the distinct elements function). Closing the gap between upper and lower bounds for the ``worst'' functions remains a major open question.
    
    We study this cost for \emph{symmetric Boolean functions}: $f:\zo^n\to\zo$ such that
    \[
    f(x_1,\ldots,x_n)=f(x_{\pi(1)},\ldots,x_{\pi(n)}) \mbox{ for any permutation }\pi:[n]\stackrel{\sim}{\to}[n].
    \]
    An equivalent way to express this condition is the following: there exists $g:\{0,\ldots,n\}\to\zo$ such that for all $x\in\zo^n$,
    \[
        f(x) = g(\|x\|) \mbox{ where }\|x\| = \mbox{ \# of 1's in } x.
    \]
    We focus on the following symmetric Boolean functions:
    \begin{itemize}
        \item Disjunction/OR: $\OR(x) = x_1\vee\cdots\vee x_n$ (and, similarly, Conjunction/AND).
        \item Majority: $\MAJ(x)=1$ if and only if $\|x\|\geq n/2$.
        \item Exact Majority: $\EMAJ(x)=1$ if and only if $\|x\|=\lfloor n/2\rfloor$.
        \item Modulo $q$: $\MOD_q(x)=1$ if and only if $\|x\|\equiv 0\pmod q$.
    \end{itemize}

    \subsection{Our Results}

    \paragraph*{Lower Bounds.}
        The authors in~\cite{ball2020complexity} adapted Nechiporuk's method to prove lower bounds for multiparty PSMs with one input bit per party. Their method relies on the number of distinct subfunctions obtained by restricting input variables. For symmetric Boolean functions, it yields only the trivial lower bound: $n$.

        In this paper we provide a simple combinatorial technique to prove lower bound for $\OR$ function. More specifically, we prove the following result.

        \begin{theorem}[Same as Corollary~\ref{cor:new-or-bound}]
                \label{intro-thm:or-bound}
            For any perfectly correct, support-private PSM protocol computing $\OR$, the communication complexity  $C(\OR) \geq (1-o(1))n \log n$. (The same bound holds for $\AND$.)
        \end{theorem}
        The matching upper bound of $O(n\log n)$ is known for this function~\cite{ishai2013randomization}.

        As a direct consequence we obtain lower bounds for other standard symmetric functions. The bound for $\EMAJ$ is tight, and the bound for $\MOD_q$ is tight up to a constant factor for prime $q$.

        \begin{theorem}[Same as Corollary~\ref{cor:new-maj-bound}]
        \label{intro-thm:maj-bound}
            For any perfectly correct, support-private PSM protocol computing $\MAJ$ or $\EMAJ$, the communication complexity is at least $(1-o(1))n\log{n}$. For $\MOD_q$ with $2 \leq q \leq n$, the communication complexity of such protocols satisfies  $C(\MOD_q) = \Omega(n\log{q})$, where the constant in $\Omega$ does not depend on $q$.
        \end{theorem}

The drawback of the technique we use to prove Theorems \ref{intro-thm:or-bound} and \ref{intro-thm:maj-bound} is that it cannot give lower bounds better than $\Omega(n \log n)$. This might be fine for symmetric functions, for which we know no function of higher cost, but it is not satisfactory for general functions. We introduce an alternative combinatorial technique to obtain lower bounds that does not suffer from this restriction. This technique is employed using the following framework:
    \begin{enumerate}
        \item \label{step:1} Use our combinatorial technique to prove a lower bound on the \emph{individual} message complexity $C_{\text{indiv}}(f)$. 
        \item \label{step:2} Lift the lower bound on $C_{\text{indiv}}(f)$ to a bound on the \emph{total} communication complexity $C(f)$ (see Lemma~\ref{lem:lifting} below). 
    \end{enumerate}

    Step~\ref{step:1} uses only support privacy.\footnote{Perfect privacy implies support privacy, so a lower bound under support privacy also applies to perfectly private protocols.} Section~\ref{sec:agreement} proves the ``agreement lemmas'' that supply the combinatorial bound. Section~\ref{sec:lower-bounds} translates these lemmas into individual message lower bounds, and the lifting lemma in Section~\ref{subsec:metrics-lifting} converts them into total communication lower bounds. 

    This technique proves a superlinear lower bound for the $\OR$ function and recovers the same lower bounds as above (up to a constant factor) for  $\MAJ$, $\EMAJ$ and $\MOD_q$  (the last one only for $q$ bounded away from $n$).
    
    However, the advantage of this technique over the one used to prove Theorems~\ref{intro-thm:or-bound} and \ref{intro-thm:maj-bound} is that it readily generalizes to arbitrary Boolean functions and is capable of proving nearly quadratic lower bounds. 
    In particular, we show that our technique subsumes the method based on Nechiporuk's bound used in~\cite{ball2020complexity}, which implies that this method works in the support-private setting. Previously, Beimel et al.~\cite{beimel2025weaklyprivate} showed that the method of~\cite{ball2020complexity} works for weak privacy.

        \paragraph*{Upper Bounds.}
        \label{subsec:upper-bounds}

    For perfectly private PSM protocols, we improve the known $O(n^2)$ communication bound for symmetric Boolean functions~\cite{shinagawa2023private} to $\widetilde{O}(n^{3/2})$. Section~\ref{sec:upper-bounds} proves:
               \begin{theorem}[same as Theorem~\ref{thm:upper-bound-perf-private-sym}]
        For any symmetric Boolean function $f \colon \{0,1\}^n \to \{0,1\}$ there is a perfectly correct, perfectly private PSM of total complexity $\widetilde{O}(n^{3/2})$.
    \end{theorem}

           For $\MAJ$, a combination of known results matches our lower bound up to polylogarithmic factors:

        \begin{corollary}[same as Corollary~\ref{cor:maj-upper-bound}]
        There is a perfectly correct, perfectly private simultaneous messages protocol for $\MAJ$ with total communication $O(n \log^3 n \log\log n)$.
    \end{corollary}
           
           For $\EMAJ$ and for $\MOD_q$ with prime $q$, the upper bounds are tight up to constant factors in the stated regimes; in other words, they match the lower bounds we establish in Section~\ref{subsec:lower-bounds-maj-friends}. The upper bounds use variants of the protocol in~\cite[Claim 4.5]{ishai2013randomization}, which we recall in Theorem~\ref{thm:folklore-or}.

           The universal support-private construction of~\cite{beimel2025weaklyprivate} gives $\widetilde{O}(n^2)$ communication. For the special case of symmetric Boolean functions, we obtain a quasilinear bound:

    \begin{theorem}[same as Theorem~\ref{thm:support-priv-upper-bound}]
        For any symmetric Boolean function $f \colon \{0,1\}^n \to \{0,1\}$ there is a perfectly correct, support-private PSM of total complexity $\widetilde{O}(n)$.
    \end{theorem}
    This upper bound also limits lower-bound techniques that use only support privacy: such techniques cannot yield more than quasilinear bounds for symmetric Boolean functions. The constructions and proofs of these upper bounds appear in Section~\ref{sec:upper-bounds}.

    



\subsection{Overview of our Techniques}

We provide a new general combinatorial framework that reduces proving lower bounds in the PSM model for arbitrary functions to studying a certain combinatorial object: $f$-cube systems (see Section~\ref{subsec:cube-systems}). 
The connection is simple: fixing the shared randomness of an $n$-party PSM assigns a transcript to each vertex of the Boolean input cube $\{0,1\}^n$. Vertices correspond to parties' inputs, and the transcript at that vertex consists of the messages sent with those inputs and that randomness. Iterating through choices of randomness, generates a system of such cubes.
We say that $v$ is the $S$-corner of a cube if it is the transcript produced when exactly the parties in $S$ hold a 1. 
For a bit $b \in \{0,1\}$, we say that an $S$-corner is a $b$-corner if the value of the function $f$ on the corresponding input is $b$.

Perfect correctness and support privacy impose constraints on such cube systems.
Perfect correctness forbids any transcript from appearing at both a $0$-corner and a $1$-corner. Support privacy, on the other hand, tells us that the same transcript must correspond to many $b$-corners across the cube system. Formally, if $v$ occurs at a $b$-corner, it means that for every input $S$ on which $f$ outputs $b$, there is some cube whose $S$-corner is $v$. 


\paragraph{$\Omega(n\log n)$ lower bounds.}
Let's consider the case of $\mathrm{OR}$. The bound is elementary enough that we can effectively sketch the whole proof here.

Let $Z$ be the set of 0-corners and let $U$ be the set of all transcripts that occur. Write $k_i$ for the size of party $i$'s message alphabet and $g=(\prod_i k_i)^{1/n}$ for the geometric mean of the alphabet sizes. We will bound $|Z|/|U|$ from both sides by showing
\[
\frac{1}{g}\le \frac{|Z|}{|U|}\le e^{-n/g}.
\]
Taking logarithms gives $g\ln g\ge n$. For $n\ge 3$, if $g\le n/\ln n$, then $g<n$ and $g\ln g<n$, a contradiction. Thus $g>n/\ln n$, and the sum of the logarithms of the local alphabet sizes satisfies
\[
\sum_i\log_2 k_i=n\log_2 g\ge n\log_2(n/\ln n)=(1-o(1))n\log_2 n.
\]

For the upper bound on $|Z|/|U|$, select one cube for each transcript in $Z$. The indicator function $\chi_Q$ of a cube $Q$ is $1$ at its vertices and $0$ elsewhere. Each selected cube contains just the point of $Z$ for which it was chosen, so the indicators are linearly independent over $\mathbb F_2$. Since all the indicators vanish outside $U$, we can view them as vectors with one entry per transcript in $U$. This gives the dimension bound $|Z|\le |U|$. If we hide some transcripts and the restricted indicators remain independent, the same argument bounds $|Z|$ by the smaller number of visible transcripts in $U$.

Choose a symbol $a_i$ in each coordinate and hide every point that uses any of these symbols in its corresponding coordinate. The visible vertices still determine each cube $Q$. If only one of its two symbols in coordinate $i$ is visible, the other must be $a_i$. For the dimension argument, reconstruction must also work for linear combinations of indicators. A coordinate line fixes all but one coordinate and lets the remaining coordinate range over its alphabet. Each line meets a Boolean cube in zero or two vertices, so every indicator and every linear combination of indicators has even parity on each line. If a point uses $a_i$ in coordinate $i$, parity determines its value from the values at the other points on its line in direction $i$. Each of those points uses one fewer chosen symbol, so we can recover all hidden entries by induction. A linear combination that vanishes on the visible points therefore vanishes everywhere, which preserves independence.

Choose the $a_i$ greedily. In coordinate $i$, choose the most frequent symbol among the transcripts of $U$ still visible. This hides at least a $1/k_i$ fraction of those transcripts. After all $n$ choices, at most a fraction $\prod_i(1-1/k_i)$ of $U$ remains visible. Applying the arithmetic mean-geometric mean inequality on the reciprocals $1/k_i$, we get
\[
\frac{|Z|}{|U|}\le \prod_i(1-1/k_i)\le e^{-\sum_i 1/k_i}\le e^{-n/g}.
\]

For the lower bound, support privacy lets us realize any transcript with output $1$ when only party $i$ holds a $1$. Flipping that bit with the same shared randomness gives a zero-output transcript on the same line. Thus every line in direction $i$ that meets $U$ contains a point of $Z$. These lines are disjoint and each contains at most $k_i$ points, so $|Z|/|U|\ge 1/k_i$ for every $i$. At least one $k_i$ is at most $g$, and hence,
\[
\frac{|Z|}{|U|} \ge \frac1g
\]

This completes the proof.

\paragraph{Alternative lower bound approach.}
Our alternative lower bound argument exploits the fact that when each party's alphabet is small, some transcript must occupy corners for which $f$ outputs different values.

We now sketch how to prove such lower bounds on the cube systems for symmetric functions by forcing a transcript collision. Suppose a cube has $v$ as its $S$-corner. Then each party's message at that corner is fixed. The message that each party would've sent had it held the opposite bit, however, still remains unfixed.

Now we consider two cubes, one whose $S$-corner is $v$ and one whose $T$-corner is $v$, where $T \subsetneq S$. Suppose that for each party $i$ in some set $X \subseteq T$, the message $i$ would've sent if its input bit flipped is the same in both of the two cubes. If the parties in $X$ switch from holding 1 to holding 0 in both cubes, the two cubes will produce the same transcript, as the $(S \setminus X)$-corner of the first cube equals the $(T \setminus X)$-corner of the second. If $f$ takes a different value on inputs of Hamming weight $|T| - |X|$ than it does on inputs of Hamming weight $|S| - |X|$, we have reached a contradiction: two inputs with different $f$-values that produce the same transcript. 

Our agreement lemmas guarantee us the existence of such a pair of cubes, and they are specialized based on $f$. For $f = \OR$, we take $X = T$, which means that the $(S \setminus T)$-corner of the first cube, where at least one party holds a 1, corresponds to the $\emptyset$-corner of the second, where no party holds a 1. The proof that such an $S$ and $T$ exist relies heavily on results from extremal combinatorics; in particular, it combines Higman's lemma with Ramsey's theorem. 

For $f \in \{\MAJ, \EMAJ, \MOD_q\}$, it is enough to choose $X$ so that exactly $p$ parties of $T$ lie outside $X$, where $p$ is a Hamming weight threshold value for $f$. In this case, we  
argue that if such $S$ and $T$ did not exist, we would have error-correcting codes whose parameters beat Plotkin's bound.

We also show that our framework subsumes the Nechiporuk method for PSM protocols used in~\cite{ball2020complexity}. In particular, we show that given a sufficiently large lower bound on the number of distinct subfunctions of $f$, one can use a simple counting argument instead of specialized agreement lemma to guarantee the existence of a contradictory transcript. This result applies to arbitrary Boolean functions. 

The agreement lemmas allow us to bound the size of the largest alphabet for any one party, so we use a lifting lemma (Lemma~\ref{lem:lifting}) to translate this individual bound into a bound on the total communication. The idea behind the lifting lemma is as follows: if a protocol has small total communication complexity, it cannot have too many long messages. Given a transcript with many short messages, we can consider the same problem restricted to the short message coordinates. Then, given a bound on the individual communication complexity of the restricted function, we can show that there must exist a transcript with even fewer short messages, which, after some iterations, yields a contradiction. 

\paragraph{Upper bounds.} Our general upper bounds compose an outer universal protocol with an inner protocol. The outer protocol evaluates a symmetric function from residues of the Hamming weight modulo small primes. Each inner protocol privately computes one outer message from the parties' input bits, without revealing the corresponding residue. We use the perfectly private universal protocol of~\cite{BeimelKN18} for the $\widetilde{O}(n^{3/2})$ bound and the support-private construction of~\cite{beimel2025weaklyprivate} for the $\widetilde{O}(n)$ bound. The upper bound for $\MAJ$ instead follows by converting a known permutation branching program into a PSM protocol.

\subsection{Related Work}
\paragraph{Concurrent lower bounds.}
Independently of this work, Holmgren and Wu~\cite{HolmgrenW26} proved a $\omega(n)$ lower bound on the complexity of $\OR_n$ in the perfectly correct, perfectly private setting using the techniques that is very similar to our alternative approach. Their $\OR$ argument and ours share a common combinatorial core: privacy and correctness produce a family of locally colored subsets, while Ramsey's theorem and Higman's lemma force an incompatible agreement. Their work also characterizes symmetric functions admitting constant-size local alphabets in terms of constant periodicity, extends this characterization to sufficiently small statistical error, and proves quadratic lower bounds for nonexplicit Boolean functions. Both our paper and~\cite{HolmgrenW26} obtained this result with substantial help from AI. Under the additional assumption that the support of the all-zero encoding is invariant under coordinate permutations, they obtain an $\Omega(n\log n)$ lower bound for $\OR$.

\paragraph{Earlier lower-bound methods.}
Ball, Holmgren, Ishai, Liu, and Malkin~\cite{ball2020complexity} adapted Nechiporuk's method to DREs, obtaining $\Omega(n^2/\log n)$ lower bounds for explicit Boolean functions. For symmetric Boolean functions, restrictions remain symmetric, and the subfunction-counting method gives only a trivial linear bound. Ball and Randolph~\cite{DBLP:conf/citc/BallR22} extended this approach to many-party PSMs with longer local inputs. Shinagawa and Nuida~\cite{ShinagawaNuida2024} developed a different combinatorial representation using abstract simplicial complexes, yielding explicit, nonasymptotic lower bounds for concrete functions.

\paragraph{Weak privacy.}
The qualitative requirement that an adversary cannot rule out a possible secret was studied in secret sharing by Beimel and Franklin~\cite{BeimelFranklin2007}. Beimel, Ishai, Kushilevitz, and Li~\cite{beimel2025weaklyprivate} developed a quantitative notion of $p$-weak privacy that bounds likelihood ratios and includes the equal-support notion considered here when no quantitative lower bound on $p$ is imposed. They also showed that the DRE lower-bound method of~\cite{ball2020complexity} extends to weak privacy, and constructed universal weakly private encodings. We use their construction as a black box for our support-private upper bound. Our term \emph{support privacy} emphasizes equality of transcript supports, without requiring that the probabilities of those transcripts be close.

\paragraph{Upper bounds.}
Universal PSM constructions for arbitrary functionalities include those of Beimel, Kushilevitz, and Nissim~\cite{BeimelKN18} and Assouline and Liu~\cite{assouline2021multi}. For symmetric functions with one input bit per party, Shinagawa, Eriguchi, Satake, and Nuida~\cite{shinagawa2023private} gave an $O(n^2)$ protocol based on quadratic residues, generalizing earlier quadratic-residue protocols of Feige, Kilian, and Naor and of Ishai. Eriguchi and Shinagawa~\cite{EriguchiShinagawa2025} later developed constructions using universal reconstruction, and Nuida~\cite{Nuida2026} improved the histogram encodings underlying this approach. Eriguchi~\cite{Eriguchi2026} subsequently related the communication complexities of symmetric and general functionalities in both PSM and CDS, including improved PSM upper bounds when the number of parties exceeds the local alphabet size. These works also study related models such as ad hoc and robust PSM. We focus on the regime of a growing number of parties with one input bit per party, and study both quantitative lower bounds and the effect of relaxing perfect privacy to support privacy.

\section{Preliminaries}
    Throughout, we use $[n]:=\{1,\dots,n\}$. We write $S\in\binom{[n]}{\geq k}$ when $S\subseteq[n]$ and $|S|\geq k$. For $S\subseteq[n]$, we take $\chi_S\in\zo^n$ to denote the characteristic vector of $S$: $(\chi_S)_i = 1 \iff i\in S$. Unless otherwise specified, $\log n$ denotes the base 2 logarithm of $n$. For a vector $x \in \{0, 1\}^n$, we use $\|x\|$ to denote the Hamming weight of $x$. We use $\mathbbm{1}[X]$ to denote the indicator function for the event $X$.

    Our bounds concern the following explicit functions. Here $x = (x_1, x_2, \ldots, x_n)$ where each $x_i \in \{0, 1\}$. 
    \begin{itemize}
    \item The OR function,  $\OR_n(x) := \bigvee_{i=1}^n x_i=1-(1-x_1)(1-x_2)\dots(1-x_n)$.
    \item The Majority function, $\MAJ_n(x) = \mathbbm{1} [\|x\| \geq n/2]$.
    \item The Modulo $q$ function, $\MOD_{q, n}(x) = \mathbbm{1}[\|x\| \equiv 0 \pmod{q}]$.
    \item The Exact Majority function, $\EMAJ_n(x) = \mathbbm{1}[\|x\| = \lfloor n/2\rfloor]$.
    \end{itemize}
When the size of the input is clear from context, we omit the subscript $n$ in our function names.     

    \subsection{Private Simultaneous Messages}

    \begin{definition}[PSM for $f$]
        Let $n\in \N$. Let $\cR$ be a nonempty finite set from which the shared random string is drawn uniformly. For each $i\in [n]$, let $\cM_i \subseteq \{0, 1\}^*$ be a set of possible messages (not necessarily of equal length, and possibly containing the empty message of length 0), and let $\mathsf{M}_i \colon \zo \times \cR\to\cM_i$ be the function which returns the message that party $i$ sends. Let $\mathsf{M} = (\mathsf{M}_1,\dots,\mathsf{M}_n)$. Let $\mathsf{Rec} \colon \cM_1\times \dots\times \cM_n \to \zo$ be some recovery function. To reduce notation, we occasionally will write $\mathsf{M}(x,r) := (\mathsf{M}_1(x_1,r),\dots,\mathsf{M}_n(x_n,r))$.

        The pair $(\mathsf{M}, \mathsf{Rec})$ is a \emph{PSM $(n,K)$-protocol for $f$} if $\max_i \abs{\cM_i} \leq K$ and it is perfectly correct, meaning that $\mathsf{Rec}(\mathsf{M}(x,r))=f(x)$ for every $x\in\zo^n$ and $r\in\cR$.

        The protocol is a \emph{perfect PSM} if it also satisfies \emph{perfect privacy}. For all $x,x'\in\zo^n$ such that $f(x)=f(x')$, when $r$ is drawn uniformly from $\cR$ we have
            \begin{equation}
                (\mathsf{M}_1(x_1,r),\dots,\mathsf{M}_n(x_n,r)) \overset{d}{=} (\mathsf{M}_1(x'_1,r),\dots,\mathsf{M}_n(x'_n,r)) \ .
            \end{equation}
    \end{definition}


    The authors in~\cite{beimel2025weaklyprivate} introduced a weaker notion inspired by differential privacy. A protocol $(\mathsf{M},\mathsf{Rec})$ is $p$-\emph{weakly private} if

    \begin{equation*}
        \Pr[\M(x, r) = y] \geq p \cdot \Pr[\M(x', r) = y]
    \end{equation*} for all $x, x'$ such that $f(x) = f(x')$ and all $y$. Observe that perfect privacy is $p$-weak privacy in the case that $p = 1$.  

    A protocol is \emph{weakly private} if it is $p$-weakly private for some $p>0$. For the finite protocols considered here, this is equivalent to equality of transcript supports on inputs with the same output. We call this property \emph{support privacy} and use the following equivalent formulation in our lower bounds:
    \begin{definition}[Support Privacy] Let $(\mathsf{M}, \mathsf{Rec})$ be a \emph{PSM $(n,K)$-protocol for $f$}. Let $\mathsf{M}(x, \mathcal{R})$ denote $\{\mathsf{M}(x, r) \mid r \in \mathcal{R}\}$. Then we say that $(\mathsf{M}, \mathsf{Rec})$ satisfies \emph{support privacy} if the following holds:
        \begin{itemize}
            \item \textbf{(Support Privacy)}
            For all $x,x'\in \zo^n$ such that $f(x)=f(x')$,
            $\mathsf{M}(x,\cR) = \mathsf{M}(x',\cR)$,
            or, equivalently, 
            for any $x$ and any $x'$ such that $f(x')=f(x)$ and any $r\in \cR$, there exists $r'\in \cR$ such that $\mathsf{M}(x,r) = \mathsf{M}(x',r')$.
        \end{itemize} 
    \end{definition}

    \subsection{Communication Complexity Metrics}
    \label{subsec:metrics-lifting}
        We discuss a few worst-case metrics for communication complexity in the PSM/DRE setting. Fix a privacy requirement, either support privacy or perfect privacy. In each of the three definitions below, the minimum ranges over perfectly correct protocols satisfying the chosen privacy requirement. We suppress this requirement in the notation. Relations between measures discussed below also apply to both privacy settings. All lower bounds in the paper hold for support privacy, and thus imply lower bounds for perfect privacy as well. Upper bounds are stated directly as protocols, together with the privacy notion they achieve.

        We first measure the largest local message alphabet. In an $(n,K)$-PSM protocol, every party's alphabet has size at most $K$. The \emph{individual message complexity} of a Boolean function $f:\zo^n\to\zo$ is the logarithm of the smallest such $K$:

        \[
            C_\text{indiv}(f) = \min\left\{\log K\mid\mbox{ there exists an $(n,K)$-PSM for $f$}\right\}.
        \]
        A slightly stronger metric, studied in~\cite{ball2020complexity,DBLP:conf/citc/BallR22}, considers communication complexity to be the sum of individual message complexities of all parties:
        \[
            C_\text{indiv-sum}(f) = \min\left\{\sum_{i=1}^n \log K_i \mid \mbox{ there is a PSM for $f$ where }\mathsf{M}:\zo^n \times \mathcal{R} \to [K_1]\times \cdots\times[K_n]\right\}.
        \]

        Note that these complexity measures are closed under taking a subfunction.

        For a function $f:\zo^n\to\zo$, a set $I\subseteq [n]$, and assignment $z\in\zo^n$, denote by $f|_{\bar I=z}:\zo^{|I|}\to\zo$ the restriction obtained by fixing the coordinates outside $I$ according to $z$:
        \[
            f|_{\bar{I}=z}(x)= f(u) \mbox{ where }u_i=\left\{\begin{array}{rl}
                x_i & \mbox{if }i\in I, \\
                z_i & \mbox{if }i\notin I.
            \end{array} \right.
        \]

        \begin{lemma} \label{lem:closure}
        For every function \(f:\zo^n\to\zo\), subset \(I\subseteq[n]\), and
    assignment \(z\in\zo^n\), we have
\[
    C_{\mathrm{indiv\text{-}sum}}(f)
    \ge C_{\mathrm{indiv\text{-}sum}}(f|_{\bar I=z})
    \qquad\text{and}\qquad
    C_{\mathrm{indiv}}(f)
    \ge C_{\mathrm{indiv}}(f|_{\bar I=z}).
\]
Both inequalities hold for perfect privacy and for support privacy.
    \end{lemma}

    \begin{proof}
    Fix a PSM protocol \((\mathsf{M},\mathsf{Rec})\) for \(f\) with shared randomness
\(R\). Let \(T=(M_j(z_j,R))_{j\in\bar I}\), and fix any \(t\) in the
support of \(T\). To compute \(f|_{\bar I=z}\), the players in \(I\)
draw their shared randomness from the distribution of \(R\) conditioned
on \(T=t\). Each player \(i\in I\) then sends \(M_i(x_i,R)\). The referee
inserts \(t\) in the coordinates \(\bar I\) and applies \(\mathsf{Rec}\).
The resulting full transcript is a transcript of the original protocol
on the extension \(u\) of \(x\) by \(z\), so the referee outputs
\(f(u)=f|_{\bar I=z}(x)\).

For privacy, take \(x,x'\) with \(f|_{\bar I=z}(x)=f|_{\bar I=z}(x')\).
Their extensions \(u,u'\) have the same \(f\)-value, so their original
transcript distributions are equal under perfect privacy and have equal
supports under support privacy. Conditioning both on \(T=t\) and
projecting onto \(I\) preserves the respective equality. Finally, each
player in \(I\) uses a subset of its original message alphabet. Thus both
the sum and the maximum of the logarithms of the alphabet sizes can only
decrease. Minimizing over protocols for \(f\) proves the claim.
    \end{proof}

        We also consider the worst-case total bit length of a transcript. The referee receives an ordered tuple of (clearly delineated) messages, and we only charge for bits in each message. Thus $f$ has \emph{communication complexity} at least $k$ if every PSM protocol $(\mathsf{M},\Rec)$ for $f$ has an input $x$ and randomness $r$ for which $|\mathsf{M}(x,r)|\geq k$.

        \[
            C(f) := \min_{(\mathsf{M},\Rec)\text{-PSM for }f } \max_{x,r}\sum_{i=1}^n |\mathsf{M}_i(x_i,r)|.
        \]
        Note that for both privacy definitions and any $f$, we have $C_\text{indiv-sum}(f) \geq C(f)$. Indeed, consider a protocol with individual party alphabet sizes $K_1, \dots, K_n$ attaining $C_{\mathrm{indiv\text{-}sum}}(f)$. Injectively encode each $\mathcal{M}_i$ into strings of length at most $\lfloor \log K_i \rfloor$. Injectiveness preserves correctness and privacy, and the resulting transcript has length at most $\sum_i \lfloor \log K_i \rfloor \le \sum_i \log K_i$.


        The lifting lemma below aims to convert individual message lower bounds for restrictions of $f$ into a total communication lower bound for $f$.

        \begin{lemma}
        \label{lem:lifting}
            Let $f:\zo^n\to\zo$ be such that there are integers $1\leq s\leq n$ and $B\geq0$ where
            for every
            \[
                I\subseteq [n] \mbox{ such that }|I|\ge s
            \]
            there exists $z\in\zo^n$ such that
            \[ C_\text{indiv}(f|_{\bar{I}=z})\geq B+1.\]
            Then,
            \[
                C(f)\ge (B+1)(n-s+1).
            \]
        \end{lemma}
        \begin{rem}
        \label{rem:equivalent-functions}
            Note that for any $I\subseteq[n]$ of size $s$,
            \begin{itemize}
                \item $\OR_n|_{\bar{I}=0} \equiv \OR_s$,
                \item $\MAJ_n |_{\bar{I} = z} \equiv \MAJ_s$ where $z$ has $\lceil n/2\rceil-\lceil s/2\rceil$ ones outside $I$,
                \item $\EMAJ_n |_{\bar{I} = z} \equiv \EMAJ_s$ where $z$ has $\lfloor n/2\rfloor-\lfloor s/2\rfloor$ ones outside $I$,
                \item $\MOD_{q, n}|_{\bar{I}=0} \equiv \MOD_{q, s}$ 
            \end{itemize}
            These restrictions explain how individual alphabet bounds can supply the hypotheses of the lifting lemma.
            
            In general for a symmetric function $f(x)=g(\|x\|)$, if $\va$ is a string with $a$ ones outside of $I$ then $f|_{\bar{I}=\va}(u) = g(a+\|u\|)$. In other words, in order to lower bound the communication complexity of a symmetric function, it suffices to find a (not too large) slice on which you can prove an alphabet size bound.
        \end{rem}
        \begin{rem}
        \label{rem:conditions-for-lifting-lemma}
            In the Lemma above, the bound on $f$ will follow if one only knows that $f$ has the property that there are integers $s,B$ as above such that for every $J\subseteq [n]$ with $|J|\ge s$ there is $I\subseteq J$ and $z$ such that $C_{\text{indiv}}(f|_{\bar{I}=z})\geq B+1$.
        \end{rem}

        The idea of the proof is that no short transcript can have too many long messages. Thus, given a transcript with many short messages, we can induce a PSM with the same privacy property on the function restricted to the short message coordinates. Then, given an individual message bound on PSMs for the restricted function, we can show that there must be a transcript with fewer short messages, which (after iterating) yields a contradiction.
        \begin{proof}[Proof of Lemma~\ref{lem:lifting}]
            The same argument applies to perfect privacy and support privacy, using the same privacy requirement in both complexity measures.

            Suppose for the sake of contradiction that
            \[
                C(f) < (B+1)(n-s+1).
            \]
            Let $(\M,\Rec)$ be a protocol that realizes this inequality.
            
            Say coordinate $i$ is \emph{$B$-long} in a transcript $\vv=(v_1,\ldots,v_n)$ if  $|v_i|>B$. Otherwise say that $i$ is \emph{$B$}-short for $\vv$.

            Let $\vv$ be a longest transcript, with length $|\vv|=k<(B+1)(n-s+1)$.
            Suppose $\vv$ has $\ell$ $B$-long coordinates. Observe that $\ell\le n-s$, so we have $n-\ell\ge s$ $B$-short coordinates. Let $I$ ($|I|=n-\ell$) denote the set of $B$-short coordinates. To simplify the notation assume $I = [n-\ell]$.
            
            By the condition on $f$, we know there exists a restriction $z$ such that
            \[
                C_{\text{indiv}}(f|_{\bar{I}=z})\geq B+1.
            \]
            Note that this condition implies that $\Image(f|_{\bar{I}=z})=\zo$. Hence for either $b\in\zo$, there exists $y_b$ such that $f|_{\bar{I}=z}(y_b)=b$. 
            For $y \in \{0,1\}^{n-\ell}$ define $u(y)\in\zo^n$ such that $u(y)_i=(y)_i$ if $i\in I$ and $z_i$ otherwise.
            By perfect privacy, for \emph{any} $y$ such that $f|_{\bar{I}=z}(y)=b$,
            \[\M(u(y), r) \overset{d}{=} \M(u(y_b), r). \]
            Under support privacy, these transcript supports are equal instead.
            Now, we can define a new PSM with the same privacy property for $f|_{\bar{I}=z}$, $(\M',\Rec')$.
            We use $\M_i$'s as before, but instead of using $R$ directly, we condition on hitting $\vv_{\bar{I}}$ on inputs $z_{\bar{I}}$. Namely, let
            \[
                R'= (R\mid\M_i(z_i, r) = v_i, \forall i\notin I).
            \]
            The conditioning event has positive probability: choose $b=\Rec(\vv)$ and the corresponding $y_b$ above; support privacy ensures that $\vv$ occurs on $u(y_b)$.
            Then, to encode we use the following for $i\in I$:
            \[
                \M'_i(x_i, r') = \M_i(x_i, r').
            \]
            To reconstruct from $\vv'$ we simply embed $\vv'$ into $\vv_{\bar{I}}$:
            \[
                \Rec'(\vv') = \Rec(\hat{u}(\vv')) \mbox{ where } \hat{u}(\vv')_i = \left\{\begin{array}{rl}
                    v'_i & \mbox{if }i\in I, \\
                    v_i &  \mbox{if }i\notin I.
                \end{array}\right.
            \]
            Perfect correctness and the chosen privacy property are inherited from $(\M,\Rec)$.
            
            To verify privacy, fix $y,y'$ with $f|_{\bar{I}=z}(y)=f|_{\bar{I}=z}(y')$ and fix a transcript $\vv'$. Perfect privacy gives
            \begin{align*}
               \Pr[\M'(y)=\vv']\Pr[\M_{\bar{I}}(z_{\bar{I}})=\vv_{\bar{I}}] &= \Pr[\M_I(y)=\vv'\mid \M_{\bar{I}}(z_{\bar{I}})=\vv_{\bar{I}}]\Pr[\M_{\bar{I}}(z_{\bar{I}})=\vv_{\bar{I}}]\\
               &=\Pr[\M(u(y))=\hat{u}(\vv')]\\
               &=\Pr[\M(u(y'))=\hat{u}(\vv')]\\
               &= \Pr[\M_I(y')=\vv'\mid \M_{\bar{I}}(z_{\bar{I}})=\vv_{\bar{I}}]\Pr[\M_{\bar{I}}(z_{\bar{I}})=\vv_{\bar{I}}]\\
               &= \Pr[\M'(y')=\vv']\Pr[\M_{\bar{I}}(z_{\bar{I}})=\vv_{\bar{I}}].
            \end{align*}
            Dividing by $\Pr[\M_{\bar{I}}(z_{\bar{I}})=\vv_{\bar{I}}]$ yields $\Pr[\M'(y)=\vv'] =\Pr[\M'(y')=\vv']$. For support privacy, replace equality by equivalence of positivity in this display; the conditional supports are therefore equal.
            
            However, because $(\M',\Rec')$ is a PSM protocol for $f|_{\bar{I}=z}$, it is subject to $C_\text{indiv}(f|_{\bar{I}=z})\geq B+1$. If all its messages had length at most $B$, each local alphabet would have at most $2^{B+1}-1$ symbols, a contradiction. Hence there must be some input $y^*$ and $r^*\in \Supp(R') = \Supp(R\mid\M_i(z_i, r) = v_i, \forall i\notin I)$ such that $\M'(y^*,r^*)$ has at least one $B$-long coordinate. Hence if we consider
            \[\vv^*=\M(u(y^*),r^*)\] 
            we note that because the original $B$-long coordinates are unchanged, we now have more ($>\ell$) $B$-long coordinates in total. Let $\ell^*$ be the number of $B$-long coordinates in $\vv^*$. 
            
            If $\ell^*>n-s$, then we have $|\vv^*|\ge (B+1)(n-s+1)$. 
            
            Otherwise, $\ell^*\le n-s$ and hence $\vv^*$ has at least $s$ $B$-short coordinates. In this case, we can apply the argument above again iteratively. As the number of $B$-long coordinates strictly increases each time, eventually we will produce a transcript with $>n-s$ $B$-long coordinates and reach a contradiction.
        \end{proof}

        \subsection{Some Known Protocols}

    We use the following protocols as black boxes.

    \begin{theorem}[Beimel--Kushilevitz--Nissim~\cite{BeimelKN18}] 
\label{thm:bkn-psm}
Let \(k,N \geq 2\) be integers.
Every function
\[
    f:[N]^k \longrightarrow \{0, 1\} 
\] 
admits a \(k\)-party private simultaneous messages protocol
with perfect correctness and perfect privacy, such that 
\[
    C_\text{indiv-sum}(f) = O\!\left(k^3 N^{k/2}\right).
\]
\end{theorem}

    \begin{theorem}[\cite{beimel2025weaklyprivate}]
    \label{thm:weakly-private-psm}
        Let $f:\{0, 1\}^n \rightarrow \{0, 1\}$ be a function. Then there is a support-private, perfectly correct PSM for $f$ with complexity $\widetilde{O}(n^2)$.
    \end{theorem}
    The actual result from \cite{beimel2025weaklyprivate} proves $p$-weak privacy for some small $p$, but support privacy suffices for our results.  

    \begin{theorem}[Ishai \cite{ishai2013randomization}]
    \label{thm:folklore-or}
        There exists a PSM protocol for $\OR$ with total communication complexity $C(\OR) \leq (1+o(1))n\log{n}$. 
    \end{theorem}
    We recall the protocol for completeness.
    \begin{proof}
    Let $p>n$ be the smallest prime greater than $n$. From Bertrand's postulate we have $p<2n$. All arithmetic that follows is modulo $p$. 
    \begin{enumerate}
        \item The shared randomness will consist of $r_1,\ldots,r_n$ uniformly random elements of $\bbZ_p$ such that $\sum_i r_i = 0$ and a uniformly random $s\in \bbZ_p^*$.
        \item The $i$th party, holding $x_i\in\zo$ then computes message $m_i = r_i + x_i\cdot s$.
        \item The evaluator then computes $y=\sum_i m_i$ and outputs $1$ if and only if $y\neq 0$.
    \end{enumerate} 

    When $\OR_n(x)=0$, the transcript is uniform over message tuples whose sum is zero. When $\OR_n(x)=1$, the sum is uniform over the nonzero field elements, and the transcript is uniform conditional on that sum. In either case, its distribution depends only on the output. In this protocol, each party sends a single element from $\bbZ_p$, yielding an overall communication cost of $n\lceil\log p\rceil\leq n(\log n+2)$.
    \end{proof}    

    \subsection{\texorpdfstring{$f$-Cube Systems}{f-Cube Systems}}
    \label{subsec:cube-systems}
        We represent transcript supports by a combinatorial object called an \emph{$f$-cube system}. This representation isolates the correctness and privacy properties used by our lower bounds.
        \begin{definition}[$f$-Cube System]
            Fix alphabets $\cM_1,\ldots,\cM_n$ and take 
            \[\Omega=\cM_1\times\cdots\times \cM_n.\]

            A \emph{Boolean Cube System} (with respect to $\Omega$) consists of a pair of disjoint sets
            \[
                Z_0,Z_1\subseteq \Omega
            \]
            and a nonempty collection of possibly degenerate $n$-dimensional Boolean cubes $\cQ$. A cube $Q\in \cQ$ is specified by two symbols at each coordinate $i\in[n]$
            \[
                \M_{i,0}(Q),\M_{i,1}(Q)\in\cM_i.
            \]
            The two symbols may coincide on an irrelevant coordinate; perfect correctness forces them to be distinct whenever $f$ depends on that coordinate. We denote the \emph{$S$-corner of a cube} $Q\in\cQ$ as
            \[
            Q(S)_i = \left\{\begin{array}{lr}
                \M_{i,1}(Q) & i\in S, \\
                \M_{i, 0}(Q) & i\notin S.
            \end{array} \right.
            \]
            We say $\cQ$ is an \emph{$f$-Cube System} if
            \begin{enumerate}
                \item For all $Q \in \mathcal{Q}$, $b \in \{0, 1\}$, and $S\subseteq [n]$, $f(\chi_S)=b \iff Q(S)\in Z_b$.
                \item For $b\in \{0,1\}$, $\forall \vv\in Z_b$ and $\forall S\subseteq[n]$ such that $f(\chi_S)=b$, there is a cube $Q\in\cQ$ with $S$-corner $\vv$: $Q(S)=\vv$.
            \end{enumerate}
        \end{definition}

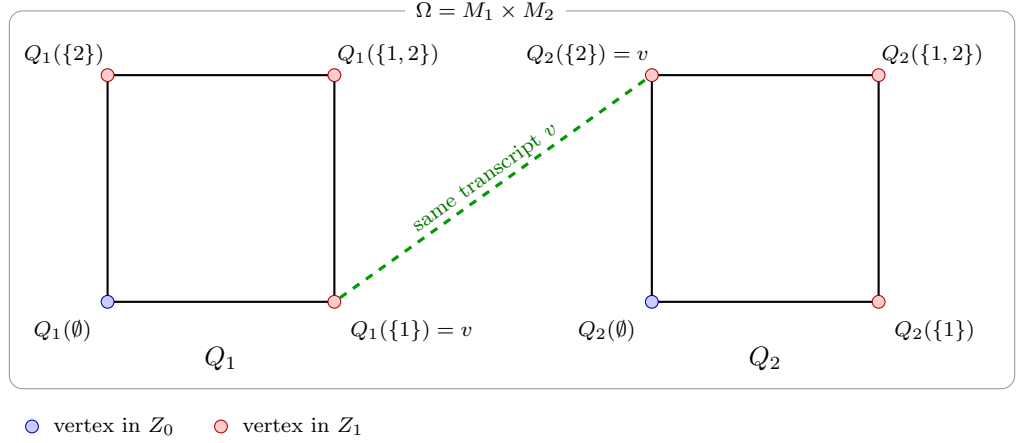
\begin{figure}[t]
    \centering
        \begin{tikzpicture}[
    font=\small,
    >=Latex,
    cubeedge/.style={thick},
    zzero/.style={
        circle,
        draw=blue!70!black,
        fill=blue!20,
        inner sep=1.7pt
    },
    zone/.style={
        circle,
        draw=red!70!black,
        fill=red!20,
        inner sep=1.7pt
    },
    lab/.style={font=\scriptsize, inner sep=1pt},
    every node/.style={align=center}
]

\begin{scope}[shift={(0,0)}]
    \coordinate (A) at (0,0);
    \coordinate (B) at (3,0);
    \coordinate (C) at (0,3);
    \coordinate (D) at (3,3);

    \draw[cubeedge] (A)--(B)--(D)--(C)--cycle;

    \node[zzero] (a1) at (A) {};
    \node[zone]  (b1) at (B) {};

    \node[
        zone,
        label={[lab,above left]{$Q_1(\{2\})$}}
    ] (c1) at (C) {};

    \node[
        zone,
        label={[lab,above right]{$Q_1(\{1,2\})$}}
    ] (d1) at (D) {};

    \node[
        lab,
        anchor=north east
    ] at ([xshift=-3pt,yshift=-4pt]a1.south west)
        {$Q_1(\emptyset)$};

    \node[
        lab,
        anchor=north west
    ] at ([xshift=3pt,yshift=-4pt]b1.south east)
        {$Q_1(\{1\}) = v$};

    \node at (1.5,-0.75) {$Q_1$};
\end{scope}

\begin{scope}[shift={(7.2,0)}]
    \coordinate (A) at (0,0);
    \coordinate (B) at (3,0);
    \coordinate (C) at (0,3);
    \coordinate (D) at (3,3);

    \draw[cubeedge] (A)--(B)--(D)--(C)--cycle;

    \node[zzero] (a2) at (A) {};
    \node[zone]  (b2) at (B) {};

    \node[
        zone,
        label={[lab,above left]{$Q_2(\{2\}) = v$}}
    ] (c2) at (C) {};

    \node[
        zone,
        label={[lab,above right]{$Q_2(\{1,2\})$}}
    ] (d2) at (D) {};

    \node[
        lab,
        anchor=north east
    ] at ([xshift=-3pt,yshift=-4pt]a2.south west)
        {$Q_2(\emptyset)$};

    \node[
        lab,
        anchor=north west
    ] at ([xshift=3pt,yshift=-4pt]b2.south east)
        {$Q_2(\{1\})$};

    \node at (1.5,-0.75) {$Q_2$};
\end{scope}

\draw[
    dashed,
    very thick,
    green!60!black
]
(b1) --
node[
    sloped,
    above,
    lab,
    text=green!45!black
]
{same transcript $v$}
(c2);

\begin{pgfonlayer}{background}
    \draw[
        black!40,
        rounded corners=5pt
    ]
    (-1.30, -1.15)
    rectangle
    (12,3.85);

    \node[
        lab,
        fill=white,
        inner xsep=4pt
    ]
    at (5.0,3.85)
    {$\Omega=M_1\times M_2$};
\end{pgfonlayer}
\node[zzero] (leg0) at (-1,-1.65) {};
\node[
    anchor=west,
    lab
] at ([xshift=0.15cm]leg0.east)
    {vertex in $Z_0$};

\node[zone] (leg1) at (1.5,-1.65) {};
\node[
    anchor=west,
    lab
] at ([xshift=0.15cm]leg1.east)
    {vertex in $Z_1$};

\end{tikzpicture}
 \caption{Illustration of an $f$-cube system for $f=\OR_2$. Vertices in $Z_b$ are transcripts decoded as $b$. Since $Q_1(\{1\})=v$ is a $1$-vertex, support privacy requires a cube $Q_2$ with $Q_2(\{2\})=v$.}
    \label{fig:f-cube-system}
\end{figure}
        
Figure~\ref{fig:f-cube-system} shows two cubes in an $\OR_2$-cube system. The next lemma states that an $f$-cube system follows from a PSM protocol for $f$.
        \begin{lemma}
            A perfectly correct support-private PSM protocol for $f$ with message space $\Omega=\cM_1\times\cdots\times\cM_n$ implies an $f$-cube system with respect to $\Omega$.
        \end{lemma}
        \begin{proof}
            
            Let $\mathcal{R}$ be the set of all possible random strings. Consider the set of cubes $\cQ=\{\Cube(r):r\in \mathcal{R}\}$, where
            \[
                \Cube(r):=\{(\M_1(x_1,r),\ldots,\M_n(x_n,r))\mid x\in\zo^n\}.
            \]
            Intuitively, the vertices of $\Cube(r)$ correspond to the possible message patterns sent by the parties when the randomness is fixed to $r$ and the input varies over all $n$-bit strings. Here each cube retains its ordered local endpoints $(\M_i(0,r),\M_i(1,r))$, even when two cubes have the same vertex set.
            
            Let $\Cube_S(r)=\M(\chi_S,r)$ denote the $S$-corner.
            
            For $b\in\zo$, define $Z_b\subseteq \Omega$ to be
            \[
                Z_b := \{(\M_1(x_1,r),\ldots,\M_n(x_n,r))\mid f(x)=b,r\in \cR\} = \{\Cube_S(r) \mid f(\chi_S)=b,r\in\cR\}.
            \]
            In other words, $Z_b$ is the set of vertices in $\Cube(r)$ corresponding to a message pattern on which $\Rec$ outputs $b$.
            We claim the result is an $f$-cube system with respect to $\Omega$.

        Note that property 1 follows by definition: for any $S\subseteq [n]$, we have $f(\chi_S)=b \implies \forall r\in\cR, \Cube_S(r)\in Z_b$. Perfect correctness implies that $Z_0$ and $Z_1$ are indeed disjoint, as required.

            Finally, note that for any $\chi_S$ and $\vv=\M(\chi_S,r)\in Z_b$ we have $f(\chi_S)=b$. By support privacy, for any $S'\subseteq[n]$ such that it is also the case that $f(\chi_{S'})=b$, there exists some choice of randomness $r'$ such that $\M(\chi_{S'},r')=\vv$. Since each cube corresponds to some choice of randomness, this is exactly property 2. 

            Hence, $\cQ,Z_0,Z_1$ is indeed an $f$-cube system with respect to $\Omega$.
        \end{proof}

\section{Lower Bounds for Symmetric Functions}
\label{sec:new-lower-bounds}

In this section we prove the following theorem.

\begin{theorem}\label{thm:new-or-bound}
$\Csum(\OR_n) \geq (1-o(1))\, n\log_2 n$.
\end{theorem}

From this theorem and from Lemma~\ref{lem:lifting} we immediately get the following corollary.

\begin{corollary}\label{cor:new-or-bound}
$C(\OR_n) \geq (1 - o(1))n\log n$.
\end{corollary}

\begin{proof}[Proof of Corollary~\ref{cor:new-or-bound}]
Recall that $\Cindiv$ is the logarithm of the largest alphabet size. It is at least the average of the logarithms of all $n$ alphabet sizes, so $\Cindiv(\OR_n) \geq \Csum(\OR_n)/n$, and from Theorem~\ref{thm:new-or-bound} we have that $\Cindiv(\OR_n) \geq (1-o(1))\log_2 n$. Now we apply Lemma~\ref{lem:lifting} with $s= n/\log n$. Fixing all variables outside of a set $I$ to $0$ turns $\OR_n$ into $\OR_{|I|}$, and for $|I|\geq s$ we have \[
\Cindiv(\OR_{|I|}) \geq (1-o(1))\log_2  s \geq (1-o(1))\log_2  n.\] 
So the lemma gives $C(\OR_n) \geq (n-s+1)\cdot(1-o(1))\log_2 n$, and the corollary follows.
\end{proof}
It is not difficult to see that the same technique gives an $(1 - o(1))n\log n$ lower bound for $\AND_n$ and $\NOR_n$ as well. 

Similarly, we can prove lower bounds for other symmetric functions.

\begin{corollary} \label{cor:new-maj-bound}
    We have the following lower bounds:
    \[
        C(\MAJ_n) \geq (1 - o(1))n\log n,\quad
        C(\EMAJ_n) \geq (1 - o(1))n\log n.
    \]
    For $2 \leq q \leq n$ we have 
    \[
    C(\MOD_q) \geq \Omega(n\log q).
    \]
\end{corollary}

\begin{proof}[Proof of Corollary~\ref{cor:new-maj-bound}]
Similarly to the previous corollary, from Theorem~\ref{thm:new-or-bound} we have that $\Cindiv(\OR_n) \geq (1-o(1))\log_2 n$. 

Note that $\MAJ_n$ contains $\OR_{n/2}$ as a subfunction, $\EMAJ_n$ contains $\AND_{n/2}$ as a subfunction, and $\MOD_q$ contains $\NOR_{q-1}$ as a subfunction. From these facts and Lemma~\ref{lem:closure}, we have
\begin{align*}
\Cindiv(\MAJ_n) &\geq (1-o(1))\log_2 n, \\
\Cindiv(\EMAJ_n) &\geq (1-o(1))\log_2 n,\\
\Cindiv(\MOD_q) &\geq \Omega(\log_2 q).
\end{align*}
Similarly to the previous corollary, using Lemma~\ref{lem:lifting}, we obtain the desired inequalities.
\end{proof}

The rest of the section is devoted to the proof of Theorem~\ref{thm:new-or-bound}.

Consider a perfectly correct, support private PSM for $\OR_n$ in which player~$i$ has an alphabet $\cM_i$ of size $k_i$. Denote its individual-sum complexity by $C$, that is,
\[
C=\sum_i \log_2 k_i .
\]
It will be convenient to denote by $q=\bigl(\prod_i k_i\bigr)^{1/n}$ the geometric mean of the alphabet sizes. In particular, $C=n\log_2 q$, so we need a lower bound on $q$.

Consider the discrete cube $K=\cM_1\times\dots\times\cM_n$ and the Boolean cube structure on it induced by the protocol. Recall that each value of the shared randomness gives a Boolean cube $Q\subseteq K$: once the randomness is fixed, player~$i$ sends one of two messages depending on the input bit, and $Q$ consists of the tuples of messages on all $2^n$ inputs. The two messages of each player are different (otherwise the all-zeros input and an input with a single one would give the same tuple, which contradicts perfect correctness), so $Q$ is a product of $n$ two-element sets. We call these cubes the cubes of the protocol. For $S\subseteq[n]$, the $S$-corner of $Q$ is the tuple of messages on the input that has ones exactly in the positions from~$S$.

Denote by $Z$ the support of all $0$-inputs in $K$ and by $O$ the support of all $1$-inputs. Denote by $U=Z\cup O$ the total support of the protocol in $K$. The only $0$-input is the all-zeros input, so $Z$ is the set of the $\emptyset$-corners of the cubes of the protocol, $O$ is the set of all their other corners, and $U$ is the union of these cubes. By perfect correctness, $O$ and $Z$ are disjoint. By support privacy, all $1$-inputs have the same support, so for every $v\in O$ and every nonempty $S\subseteq[n]$ there is a cube of the protocol whose $S$-corner is~$v$.

The plan of the proof is simple: we bound the fraction of $U$ occupied by $Z$ from above (this is the main part) and from below, and then compare the two bounds.

We start with the upper bound. Among all Boolean cubes of the protocol pick, for each $z\in Z$, one cube $Q_z$ that contains $z$, and denote by $\cF$ the family of the picked cubes. Each cube of the protocol contains exactly one point of $Z$, namely its $\emptyset$-corner (the other corners are in $O$), so the cubes $Q_z$ are different for different $z$. In particular, $|\cF|=|Z|$. For each cube $Q$ consider its characteristic function $\chi_Q\colon K\to\F_2$ that for any point of $K$ outputs $1$ iff the point is in $Q$ (sometimes it is better to think of $\chi_Q$ as a $|K|$-dimensional vector over $\F_2$). Note that the functions $\chi_Q$ for $Q\in\cF$ are linearly independent. Indeed, their restrictions to $Z$ are the indicators of single points: $\chi_{Q_z}$ is equal to $1$ on $z$ and to $0$ on the rest of $Z$.

All these functions are zero outside of $U$, so for now we only get the trivial bound $|Z|=|\cF|\leq|U|$. The next lemma strengthens this observation and is the main (and maybe the only) ingredient of the proof: the functions remain independent even if we forget their values on a large part of $K$.

\begin{lemma}\label{lem:main}
For any choice of $a_i\in\cM_i$ for each $i$ denote $a=(a_1,\dots,a_n)$ and
\[
K_a=(\cM_1\setminus\{a_1\})\times\dots\times(\cM_n\setminus\{a_n\}),
\]
that is, $K_a$ consists of the points of $K$ that differ from $a$ in all coordinates. Then for every $a$ we have $|\cF|=|Z|\leq|U\cap K_a|$. As a consequence,
\begin{equation}\label{eq:upper}
|Z|\leq|U|\prod_i\left(1-\frac{1}{k_i}\right).
\end{equation}
\end{lemma}

\begin{proof}
Consider the discrete cube $K$. A \emph{line} in $K$ is a set of points in which all coordinates but one are fixed and the unfixed coordinate ranges over all possible values. If the unfixed coordinate is the $i$-th one, we say that this is a line in direction~$i$.

Note that for any Boolean cube $Q$ the function $\chi_Q$ has an even number of ones on each line (more precisely, either $0$ or $2$). Indeed, $Q$ is a product of two-element sets, so a line in direction $i$ that intersects $Q$ contains exactly two points of $Q$, one for each of the two possible values of the $i$-th coordinate. Thus the sum of $\chi_Q$ over each line is $0$.

The observation behind the lemma is that the values of $\chi_Q$ outside of $K_a$ can be reconstructed from its values on $K_a$. More precisely, the value of $\chi_Q$ at any point of $K$ is a linear combination (over $\F_2$) of its values at some points of $K_a$, and this combination is the same for all cubes $Q$. To see this, say that the $i$-th coordinate of a point $v$ is \emph{bad} if $v_i=a_i$, so that $K_a$ consists of the points with no bad coordinates. We use induction on the number of bad coordinates of a point $v$. If there are none, there is nothing to prove. If the $i$-th coordinate of $v$ is bad, consider the line through $v$ in direction~$i$. The value of $\chi_Q$ at $v$ is equal to the sum of its values at all other points of this line, and each of these points has one bad coordinate fewer than $v$.

From the observation we can conclude that the functions $\chi_Q$ for $Q\in\cF$ are independent even when restricted to $K_a$.\footnote{If this is not intuitive, think of the functions $\chi_Q$ as column vectors organized in a matrix; the rows of the matrix correspond to the points of $K$. We noted above that the columns of the matrix are independent, and thus the matrix rank is equal to the number of columns. The observation gave us that each row corresponding to a point outside of $K_a$ can be expressed as a linear combination of the rows corresponding to points of $K_a$. Removing these rows does not change the rank, so the columns are still independent.} Indeed, if some linear combination of these functions is zero on $K_a$, then by the observation it is zero on the whole $K$, and this is possible only for the trivial combination. All cubes of the protocol are inside $U$, so the restricted functions can be nonzero only on $U\cap K_a$. This means that we have $|\cF|$ independent vectors in a space of dimension $|U\cap K_a|$, and thus $|Z|=|\cF|\leq|U\cap K_a|$.

To prove the second inequality we need to find $a$ for which $U\cap K_a$ is small. We can just pick $a_1,\dots,a_n$ one by one in a greedy way, each time removing the maximal possible number of the remaining points from the support. Here the remaining points are the points of $U$ that differ from $a_1,\dots,a_{i-1}$ in the first $i-1$ coordinates, and we pick as $a_i$ the most popular value of the $i$-th coordinate among them. This removes at least a $1/k_i$ fraction of the remaining points. In the end the remaining points are exactly the points of $U\cap K_a$, and there are at most $|U|\prod_i(1-1/k_i)$ of them. Alternatively, we can pick $a$ uniformly at random and use the expectation argument: each point of $U$ gets into $K_a$ with probability $\prod_i(1-1/k_i)$, so for some $a$ the size of $U\cap K_a$ is at most its expectation $|U|\prod_i(1-1/k_i)$.
\end{proof}

For the lower bound, we need one more simple observation. Consider the whole cube system and consider any line in direction $i$ in $K$. If it contains a point of $U$, it must contain a point of $Z$ too. This follows from support privacy: if the line contains a point $v\in O$, then there is a cube $Q$ of the protocol whose $\{i\}$-corner is $v$. The $\emptyset$-corner of the same cube $Q$ differs from $v$ only in the $i$-th coordinate, so it is on the same line, and it is in $Z$.

Now we can count the number of lines. The lines in direction $i$ are disjoint and each of them consists of $k_i$ points. Each line that contains a point of $U$ also contains a point of $Z$, so there are at most $|Z|$ such lines, and together they have at most $k_i|Z|$ points. Thus $|Z|\geq|U|/k_i$ for every~$i$. Denoting $k=\min_i k_i$ we get
\begin{equation}\label{eq:lower}
|U|\leq k|Z|.
\end{equation}

Combining \eqref{eq:upper} and \eqref{eq:lower} we get
\[
|U|\leq k|Z|\leq k|U|\prod_i\left(1-\frac{1}{k_i}\right).
\]
Canceling out $|U|$ and using $1+x\leq e^x$ for $x=-1/k_i$ we get
\[
1\leq k\prod_i\left(1-\frac{1}{k_i}\right)\leq\exp\left(\ln k-\sum_i\frac{1}{k_i}\right).
\]
Taking logarithms of both sides and using AM-GM on $\frac{1}{k_1},\dots,\frac{1}{k_n}$ we get
\[
\ln k\geq\sum_i\frac{1}{k_i}\geq n\left(\prod_i\frac{1}{k_i}\right)^{1/n}=\frac{n}{q},
\]
where recall that $q=\bigl(\prod_i k_i\bigr)^{1/n}$. Observing that $k\leq q$ (the first one is the minimum of the $k_i$'s, the second one is their geometric mean) we get
\begin{equation}\label{eq:q}
n\leq q\ln q.
\end{equation}
From this, by the standard argument, we can deduce that $q>n/\ln n$. Indeed, the function $x\ln x$ is increasing for $x\geq1$, so if $q\leq n/\ln n$, then $q\ln q\leq(n/\ln n)\cdot\ln(n/\ln n)<n$ for $n\geq3$, which contradicts \eqref{eq:q}. Thus
\[
C=n\log_2 q\geq n\log_2\frac{n}{\ln n}=n\log_2 n-n\log_2\ln n=(1-o(1))\,n\log_2 n,
\]
and Theorem~\ref{thm:new-or-bound} is proved.

\subsection{Matching Upper Bounds}

For $\OR$ the matching upper bound was shown by Ishai~\cite{ishai2013randomization} and is provided in Theorem~\ref{thm:folklore-or} above.

For $\EMAJ$ and for $\MOD_q$ with prime $q$, matching upper bounds follow from variants of the $\OR$ protocol in Theorem~\ref{thm:folklore-or}. We include the constructions for completeness.

\begin{theorem} If $q$ is prime, there exists a perfectly correct, perfectly private PSM protocol for $\MOD_{q, n}$ of cost $n\lceil \log{q}\rceil$.
\end{theorem}
\begin{proof}
   All arithmetic that follows is modulo $q$. Consider the following protocol. 
    \begin{enumerate}
        \item The shared randomness will consist of $r_1,\ldots,r_n$ uniformly random elements of $\bbZ_q$ such that $\sum_i r_i = 0$ and a uniformly random $s\in \bbZ_q^*$.
        \item The $i$th party, holding $x_i\in\zo$ then computes message $m_i = r_i + x_i\cdot s$.
        \item The evaluator then computes $y=\sum_i m_i$ and outputs $1$ if and only if $y = 0$.
    \end{enumerate} 

    It is not difficult to see that if $\sum_{i = 1}^n x_i \equiv 0 \pmod{q}$, then the evaluator sees uniformly random messages that sum to 0. On the other hand, if $\sum_{i = 1}^n x_i \not\equiv 0 \pmod{q}$, then the evaluator sees uniformly random messages summing to a uniformly random non-zero element. In this protocol, each party sends a single element from $\bbZ_q$, yielding an overall communication cost of $n\lceil\log q\rceil$.
    \end{proof}    

\begin{theorem} There exists a perfectly correct, perfectly private PSM protocol for $\EMAJ_n$ of cost $(1+o(1))n\log{n}$.
\end{theorem}
\begin{proof}
   Let $p > n$ be the smallest prime greater than $n$. All arithmetic that follows is modulo $p$. Consider the following protocol. 
    \begin{enumerate}
        \item Sample $s$ uniformly from $\bbZ_p^*$. Then sample $(r_1,\ldots,r_n)$ uniformly from the tuples in $\bbZ_p^n$ satisfying $\sum_i r_i=-s\lfloor n/2\rfloor$.
        \item The $i$th party, holding $x_i\in\zo$ then computes message $m_i = r_i + x_i\cdot s$.
        \item The evaluator then computes $y=\sum_i m_i$ and outputs $1$ if and only if $y = 0$.
    \end{enumerate} 

    The message sum is $s(\|x\|-\lfloor n/2\rfloor)$. Since $p>n$ and $s\neq0$, this sum is zero if and only if $\EMAJ_n(x)=1$. In that case, the transcript is uniform over zero-sum tuples. Otherwise, its sum is uniform over $\bbZ_p^*$, and the transcript is uniform conditional on that sum. Its distribution therefore depends only on the output. Each party sends one element of $\bbZ_p$. By Bertrand's postulate $p< 2n$, so the total communication is $(1+o(1))n\log n$.
    \end{proof}  

\section{Agreement Lemmas}
\label{sec:agreement}

In this section we provide combinatorial framework needed for our alternative approach to lower bounds.
The agreement lemmas below show that small local alphabets force compatible colorings on nested sets. Section~\ref{sec:lower-bounds} uses these agreements to produce transcripts with conflicting outputs thus resulting in the lower bounds on communication complexity in PSM setting.

Although the bounds we prove for symmetric functions are not stronger than the bounds obtained by the other approach (and for $\OR$ they are weaker), we still provide them, since the technique might be of independent interest.

\subsection{Agreement Lemma for OR}
The combinatorial lemma we need for $\OR$ says that for any finite set of colors, $\Sigma$, there is an integer $N$ such that if for each $S\subseteq [N]$, we define a coloring of its members using $\Sigma$, then there must exist non-empty $A,B\subseteq[N]$ such that $A$ is a strict subset of $B$ and the $B$-coloring agrees with the $A$-coloring on all members of $A$.

\begin{lemma}[Agreement Lemma for $\OR$]\label{lem:match}
    For every finite alphabet $\Sigma$ there is an integer $N(\Sigma)$, depending only on $|\Sigma|$, such that the following holds for every $n \ge N(\Sigma)$.
    For every family of functions $\{\alpha_S : S \to \Sigma\}$ indexed by the
nonempty sets $S \subseteq [n]$, there exist nonempty sets $T \subsetneq S
\subseteq [n]$ with $\alpha_T = \alpha_S|_T$.
    \end{lemma}
    
    The proof of this lemma in turn relies on variants of two classic results. We present derivations of these variations in Appendix~\ref{sec:combinatorial} for completeness.
    \begin{lemma}[Simplified version of Higman's Lemma]\label{lem:Higman+}
        Let $\Sigma$ be a finite alphabet. Then there is an integer $m=m(\Sigma)$ such that for every sequence of strings over $\Sigma$,
        \[w_1,\ldots,w_m \mbox{ with } w_i \in \Sigma^i \text{ for each } i \in [m]\]
        there must be some $i< j$ such that $w_i$ is a subsequence of $w_j$.  
    \end{lemma} 
    \begin{theorem}[Iterated Finite version of Ramsey's Theorem]\label{thm:Ramsey+}
        Let $\Sigma$ be a finite alphabet and $m\ge 1$. Then there is an integer $N(\Sigma,m)$ such that the following holds:

        If for each integer $k\in[1,m]$, every $k$-element subset of an $N(\Sigma,m)$-element set is colored by a string in $\Sigma^k$, then there is an $m$-element subset $H$ such that for each $k\in[1,m]$, all $k$-element subsets of $H$ have the same $\Sigma^k$-color.
    \end{theorem}
    \begin{proof}[Proof of Lemma~\ref{lem:match}]
        Let $m=m(\Sigma)$ be given by Lemma~\ref{lem:Higman+}, and set $N(\Sigma)=N(\Sigma,m)$ using Theorem~\ref{thm:Ramsey+}.

        For any $k$-set $A=\{a_1<\cdots<a_k\}$ encode its local coloring by the ordered string
        \[
            \chi_k(A) = \left(\alpha_A(a_1),\ldots,\alpha_A(a_k)\right)\in\Sigma^k.
        \]
        Now by Theorem~\ref{thm:Ramsey+} there is an $m$-element set $B=\{b_1<\cdots<b_m\}$ on which for all $k\in[1,m]$, $\chi_k$ is constant; in other words, for each $k$, all $k$-subsets of $B$ received the same $\Sigma^k$-color. For each $k$, define $v_k\in\Sigma^k$ to be the common color of $\chi_k$. This gives a sequence of strings
        \[
            v_1,v_2,\ldots,v_m \mbox{ such that }\forall i, v_i\in\Sigma^i.
        \]
        In particular, $v_i=(\alpha_{B_i}(b_1),\ldots,\alpha_{B_i}(b_i))$ for $B_i=\{b_1,\ldots,b_i\}$.

        By Lemma~\ref{lem:Higman+}, there are $s,t\in[1,m]$ such that $t<s$ and $v_t$ is a subsequence of $v_s$. Let $1\le j_1<\cdots <j_t\le s$ be the indices of $v_s$ that realize the subsequence relation ($\forall i\in [t], (v_t)_i = (v_s)_{j_i}$). Consider the two sets
        \begin{align*} S &:= \{b_1,\ldots,b_s\}&\mbox{ (the first $s$ elements of $B$)},\\
        T&:= \{b_{j_1},\ldots,b_{j_t}\} &\mbox{ (the corresponding subseq.~from $S$).}
        \end{align*}
        By homogeneity, the color of $T$ is $v_t$ and the color of $S$ is $v_s$. By the subsequence relation, these colorings agree on the elements of $T$, or, in other words,
        \[
            \alpha_T = \alpha_S|_T.
        \]
    \end{proof}

\subsection{\texorpdfstring{Agreement Lemma for $\MAJ$ and Friends}{Agreement Lemma for MAJ and Friends}}
For majority, we only need approximate agreement. The combinatorial lemma we need for the majority function says that, for any $n$, if the number of colors is at most $c_0\sqrt{n}$ for a sufficiently small absolute constant $c_0>0$ and $n$ is sufficiently large, then for any family of colorings $\alpha_S:S\to\Sigma$ with $|\Sigma|=k$ and $S\in\binom{[n]}{\geq\lceil n/2\rceil}$, there must exist nonempty sets $A,B\in\binom{[n]}{\geq\lceil n/2\rceil}$ such that $A$ is a strict subset of $B$ and the $B$-coloring disagrees with the $A$-coloring on at most $p=\lceil n/2\rceil-1$ members of $A$.

\begin{lemma}[Agreement Lemma for Majority]\label{lem:match-maj}
        Consider sufficiently large $n$ and denote $p = \lceil \frac n2 \rceil -1$. Suppose for any alphabet $\Sigma$ with $|\Sigma|\leq c_0\sqrt n$ for a sufficiently small absolute constant $c_0>0$ we are given a family of functions
        \[
            \alpha_S:S\to\Sigma,
        \]
        where $S\in \binom{[n]}{\geq p+1}$. Then there exist non-empty sets $T\in \binom{[n]}{\geq p+1}$ and $S\in \binom{[n]}{\geq p+1}$ such that
        \begin{enumerate}
            \item $T\subsetneq S$,
            \item the functions $\alpha_T$ and $\alpha_S|_T$ disagree on at most $p$ points.
        \end{enumerate}
    \end{lemma}

    The proof of the lemma relies on Plotkin's bound for error-correcting codes. In it, we argue that if such sets $T$ and $S$ did not exist, we could use the $\alpha_S$'s to construct an error-correcting code whose efficiency violates Plotkin's bound when the alphabet size is small. 

    \begin{lemma}[Plotkin's bound (for $k$-ary codes)~\cite{Quistorff03}]
        For a code $C$ with alphabet size $k$, codeword length $n$, and distance $d > \left(1 - \frac 1k\right)n$, we have that $|C| \leq \frac{kd}{kd - (k-1)n}$. 
    \end{lemma}
    
    \begin{proof}[Proof of Lemma~\ref{lem:match-maj}]
        Assume that such a pair of sets $S$ and $T$ does not exist and let $k \coloneq |\Sigma|$.
        Let $\ell=\lfloor\frac{p}{2k}\rfloor$ and for all $i = 1, \ldots, \ell$ let $S_i = [p+i]$. Consider the functions $\alpha_{S_i}|_{[p]}$ for all $i$ as vectors $v_i = (\alpha_{S_i}(1), \ldots, \alpha_{S_i}(p)) \in \Sigma^p$. By our assumption, any two of the functions $\alpha_{S_i}$ differ on at least $p+1$ coordinates. Since the compared functions have at most $\ell$ coordinates outside $[p]$, the vectors $v_i$ differ in at least $p+1-\ell$ coordinates. Thus, $\{v_i\}$ forms a code of size $\ell$, codeword length $p$, and distance $d \geq p+1-\ell \geq (1 - \frac{1}{2k})p$.

        Then, by Plotkin's bound and the fact that distance $d$ cannot exceed codeword length $p$, we have
        $$
        \ell \leq \frac{kd}{kd - (k-1)p} \leq \frac{kd}{k(1 - \frac{1}{2k})p - (k-1)p} = \frac{kd}{\frac{p}{2}} \leq 2k.
        $$
        Thus, by definition of $\ell$ we can write 
        $$
        \frac{p}{2k}-1 < \ell \leq 2k
        $$
        and conclude that
        $$
        k = \Omega(\sqrt{n}),
        $$
        contradicting $k\leq c_0\sqrt n$ when $c_0$ is sufficiently small.
    \end{proof}

    The same code argument gives the following generalization.

    \begin{lemma}[Generalized Agreement Lemma]\label{lem:match-generalized}
        Suppose for integers $n$ and $1\leq p\leq n$ and an alphabet $\Sigma$ we are given a family of functions
        \[
            \alpha_S:S\to\Sigma,
        \]
        where $S\subseteq [n]$. Consider a sequence of sets 
        \[
        S_1 \subsetneq S_2 \subsetneq \ldots \subsetneq S_\ell,
        \]
        with $\ell \geq a\sqrt{p}$ for a fixed constant $a>0$, such that $|S_i|=p+i$.
        Suppose that for all $i<j$ the functions $\alpha_{S_i}$ and $\alpha_{S_j}|_{S_i}$ disagree on more than $p$ points. Then $|\Sigma| = \Omega_a(\sqrt{p})$.
    \end{lemma}

    \begin{proof}
        Write $k=|\Sigma|$ and suppose it is a sufficiently small constant times $\sqrt p$. Then, for sufficiently large $p$, we have
        \[
            t=\min\{l,\lfloor p/(2k)\rfloor\}>2k.
        \]
        Apply the code argument from Lemma~\ref{lem:match-maj} to the first $t$ sets of the sequence $S_1,\ldots,S_l$, restricting the colorings to a common $p$-element set. It yields $t\leq2k$, a contradiction, and hence the claimed alphabet lower bound.
    \end{proof}

    We use this generalized agreement lemma for $\EMAJ$ and $\MOD_q$.

\section{Alternative Approach to Lower Bounds}
\label{sec:lower-bounds}

We use the cube-system representation from Section~\ref{subsec:cube-systems}. For a fixed shared random string $r\in\cR$, the cube $\Cube(r)$ has vertex set
\[
    \{\mathsf{M}(x,r)\mid x\in\zo^n\}
    \subseteq\cM_1\times\cdots\times\cM_n.
\]
For $S\subseteq[n]$, its $S$-corner is $\Cube_S(r)=\mathsf{M}(\chi_S,r)$. Throughout this section, \emph{cubes} refer to the members of $\{\Cube(r)\mid r\in\cR\}$.

\subsection{A Superlinear Lower Bound for OR}






    We devote this section to proving the following theorem, which gives a weaker version of Corollary~\ref{cor:new-or-bound}.
    
    \begin{theorem}\label{thm:mainResultFromWeakerAssumptions}
        For any constant $K$ there exists $N$ such that for all $n\ge N$, no $(n,K)$-protocol $(\mathsf{M},\mathsf{Rec})$ satisfies perfect correctness and support privacy for $\OR$.
    \end{theorem}



    The proof of Theorem~\ref{thm:mainResultFromWeakerAssumptions} uses the cube-system representation of perfectly correct, support-private PSMs developed in Section~\ref{subsec:cube-systems}. Perfect correctness implies that in any given cube $\Cube(r)$ in a cube system $\mathcal{Q}$ for $\OR$, exactly one vertex $v \in \Cube(r)$ will be such that $\mathsf{Rec}(v) = 0$. We refer to the unique vertex in a cube that causes the recovery function $\mathsf{Rec}$ to output 0 as the \emph{root} of that cube, and our argument exploits constraints on which positions roots may occupy in cubes in $\mathcal{Q}$. As an initial observation, we note that if some $v \in \Cube(r)$ is a root and $v$ is also a vertex in some other cube $\Cube(r')$, then $v$ must also be a root in $\Cube(r')$. Indeed, if it were not, $\Rec(v)$ would output a different value depending on whether the randomness is $r$ or $r'$, which violates perfect correctness. By the same reasoning, if $v \in \Cube(r)$ is a non-root, then $v$ is a non-root in all other cubes in which it appears.
    

    Support privacy implies that for any $b,b'\in \zo^n \setminus \{0^n\}$ and any $r\in \cR$, there exists $r'\in \cR$ such that $\mathsf{M}(b,r) = \mathsf{M}(b',r')$. In the language of cubes, each non-root $v$ occurs as an $S$-corner for every nonempty $S\subseteq[n]$.


    For example, $\{\Cube_\emptyset(r)\mid r\in\cR\}$ is the set of all roots in the cube system $\mathcal{Q}$.
    
    \begin{proof}[Proof of Theorem~\ref{thm:mainResultFromWeakerAssumptions}]
    Correctness is clearly infeasible for $K=1$, so fix any $K\ge 2$. Let $N(K)$ be as promised in Lemma~\ref{lem:match}.
    
    Assume for the sake of contradiction there exists a perfectly correct support-private PSM protocol for $\OR_n$ for $n\ge N(K)$ with $|\cM_i|\le K$.
    
    Fix some $r^*\in\cR$ and consider $\Cube(r^*)$. Let $v$ be the particular non-root
    \[
        v:=\Cube_{[n]}(r^*).
    \]
    For every nonempty $S\subseteq [n]$, support privacy implies that there exists $r_S$ (and an associated cube $\Cube(r_S)$) such that
    \[\Cube_S(r_S) = v.\]
    Take $\gamma^S := \Cube_\emptyset(r_S)$ to be the root of this cube.
    Note that we have 
    \[ \gamma^S_i = v_i \iff i\notin S.\]
    Now for each $i\in[n]$, fix any injective (partial) coloring
    \[
        \lambda_{i}:\cM_i\setminus \{v_i\} \to [K-1].
    \]
    Finally, define a coloring of each nonempty $S$ as follows:
    \[
       \text{for }i\in S,\qquad  \alpha_S(i) = \lambda_i(\gamma^S_i).
    \]
    This is well defined because of our observation that $\gamma^S_i = v_i \iff i \notin S$.
    
    Applying Lemma~\ref{lem:match}, we get that there is a particular $S$ and nonempty $T$ such that $T\subsetneq S$ and
    \[ \alpha_T = \alpha_S|_T.\]
    By the injectivity of $\lambda_i$'s, this implies
    \[ \forall i\in T,\quad \gamma^T_i = \gamma^S_i.\]
    Now consider the $S\setminus T$ corner of $\Cube(r_S)$, namely $\Cube_{S\setminus T}(r_S)$. Because $S\setminus T$ is non-empty, this is a non-root; that is, $\Rec(\Cube_{S\setminus T}(r_S))=1$. On the other hand, we will show that $\gamma^T=\Cube_{S\setminus T}(r_S)$, hence $\Cube_{S\setminus T}(r_S)$ is also a root, which contradicts perfect correctness.
    
    To see this final claim, we give a case analysis:
    \begin{itemize}
        \item If $i\in T$, then $i\notin S\setminus T$ and $\Cube_{S\setminus T}(r_S)_i=\gamma^S_i$ but also for $i\in T$ we have $\gamma^S_i=\gamma^T_i$. Hence, $\Cube_{S\setminus T}(r_S)_i=\gamma^T_i$.
        \item If $i\in S\setminus T$, then $\Cube_{S\setminus T}(r_S)_i=v_i$, but because $i\notin T$ we also have $\gamma^T_i=v_i$.
        \item Finally, if $i \notin S$, then $i\notin T$ and $\Cube_{S\setminus T}(r_S)_i=\gamma^S_i = v_i = \gamma^T_i$.
    \end{itemize}
    \end{proof}

Observe that by Remark~\ref{rem:equivalent-functions} and~\ref{rem:conditions-for-lifting-lemma}, the hypothesis of Lemma~\ref{lem:lifting} is satisfied.  Applying Lemma~\ref{lem:lifting} with $s=\lceil n/2\rceil$ and any fixed integer $B$, we lift the lower bound on the individual message complexity $ C_{\text{indiv}}(\OR)$ given by Theorem~\ref{thm:mainResultFromWeakerAssumptions} to obtain the following lower bound on the PSM communication complexity $C(\OR)$:
\begin{corollary}
\label{cor:or-lower-bound}
    The PSM communication complexity $C(\OR) = \omega(n)$.
\end{corollary}

\begin{remark}
    The agreement-lemma proof gives only $\omega(n)$ lower bound for $\OR_n$. We remark that this approach cannot give a substantially better bound. To see why, let $m \coloneq m(\Sigma)$ denote the smallest integer such that the condition of Higman's lemma holds; that is, $m$ is the smallest integer such that for every sequence of strings over $\Sigma$,
        \[w_1,\ldots,w_m \mbox{ with } w_i \in \Sigma^i \text{ for each } i \in [m]\]
        there must be some $i< j$ such that $w_i$ is a subsequence of $w_j$. It is known that $m(\Sigma)$ grows faster than any primitive recursive function of $|\Sigma|$~\cite{Friedman2001LongFiniteSequences}. 
        
As a result, it is natural to ask if we could prove a substantially sharper version of Lemma~\ref{lem:match} directly, without relying on Higman's Lemma. It is easy to see that this is not the case. Call a sequence $w_1,\ldots,w_m \mbox{ with } w_i \in \Sigma^i$  bad if it does not have a pair $w_i, w_j$ such that $w_i$ is a subsequence of $w_j$. Any bad sequence induces a coloring function violating Lemma~\ref{lem:match}. Indeed, given a bad sequence $w_1,\ldots,w_m$, if $S = \{s_1 < \ldots < s_k\} \subseteq [m]$, then we can define $\alpha_S(s_j) = (w_k)_j$. Then any $T \subsetneq S$ such that $\alpha_T = \alpha_S$ on $T$ would contradict that $w_1,\ldots,w_m$ was bad. Hence, the threshold $N(\Sigma)$ in Lemma~\ref{lem:match} also necessarily grows rapidly with $|\Sigma|$. 
\end{remark}

\subsection{\texorpdfstring{Lower Bounds for $\MAJ$ and Friends}{Lower Bounds for MAJ and Friends}}
\label{subsec:lower-bounds-maj-friends}

In this section, we prove the following theorem, which provides an alternative proof of the lower bound for $\MAJ$ function.



    
    \begin{theorem}\label{thm:mainResultFromWeakerAssumptions-maj}
        Any $(n,K)$-protocol $(\mathsf{M}, \mathsf{Rec})$ that satisfies perfect correctness and support privacy for $\MAJ$ must have $K = \Omega(\sqrt{n})$.
    \end{theorem}

    For simplicity of notation we only consider the case of odd $n = 2p+1$. The proof readily generalizes to arbitrary $n$. Indeed, fixing one input of $\MAJ_{2p}$ to zero gives $\MAJ_{2p-1}$; conditioning on an occurring message of the fixed party preserves support privacy and the local alphabets.
    
    For a cube $\Cube(r)$ and an input $x$ with $\|x\|\leq p$, consider the vertex $v=\mathsf{M}(x,r)$. By perfect correctness we have $\mathsf{Rec}(v) = 0$. In this case we call $v$ a \emph{0-vertex} of a cube. Similarly, for $x$ with $\|x\|\geq p+1$ and $v = \mathsf{M}(x,r)$ we have $\mathsf{Rec}(v) = 1$. These vertices are called \emph{$1$-vertices} of the cube. It is still convenient to denote by the \emph{root} of $\Cube(r)$ the vertex $v = \mathsf{M}(0^n,r)$.
    
    Support privacy implies that for any $b,b'\in \zo^n_{\geq p+1}$ and any $r\in \cR$, there exists $r'\in \cR$ such that $\mathsf{M}(b,r) = \mathsf{M}(b',r')$. Phrased in the language of cubes, if $v\in \Cube(r)$ is some 1-vertex of cube $\Cube(r)$, for every other possible position of a 1-vertex in a cube, there must exist some other $\Cube(r')$ so that $v$ has that position within the cube $\Cube(r')$. 

    \begin{proof}[Proof of Theorem~\ref{thm:mainResultFromWeakerAssumptions-maj}]
     
    Consider a perfectly correct support-private PSM for $\MAJ_n$ with $|\cM_i|\le K$.
    
    Fix some $r^*\in\cR$ and consider $\Cube(r^*)$. Let $v$ be the vertex that is opposite to the root:
    \[
        v:=\Cube_{[n]}(r^*).
    \]
    For every $S\in \binom{[n]}{\geq p+1}$, support privacy implies there exists $r_S$ (and an associated cube $\Cube(r_S)$) such that
    \[\Cube_S(r_S) = v.\]
    Take $\gamma^S := \Cube_\emptyset(r_S)$ to be the root of this cube.
    Note that we have 
    \[ \gamma^S_i = v_i \iff i\notin S.\]
    Now for each $i\in[n]$, fix any injective (partial) coloring
    \[
        \lambda_{i}:\cM_i\setminus \{v_i\} \to [K-1].
    \]
    Finally, define a coloring of each $S\in \binom{[n]}{\geq p+1}$ as follows:
    \[
       \text{for }i\in S,\qquad  \alpha_S(i) = \lambda_i(\gamma^S_i).
    \]
    (This is well defined because of our observation that $\gamma^S_i = v_i \iff i \notin S$.)
    
    Assuming $K \leq  C \sqrt{n}$ for small enough constant $C$ and applying Lemma~\ref{lem:match-maj} we get that there are sets $S,T\subseteq[n]$, each of size at least $p+1$, such that $T\subsetneq S$ and
    functions $\alpha_T$ and $\alpha_S|_T$ disagree on at most $p$ inputs.
    By the injectivity of the $\lambda_i$'s, this implies that
    \[
        \gamma^T_i\neq\gamma^S_i
    \]
    on at most $p$ coordinates $i\in T$. Equivalently, the roots $\gamma^T$ and $\gamma^S$ disagree on at most $p$ coordinates of $T$.

    Consider a subset $X \subseteq T$ such that $|T\setminus X|=p$ and $\gamma^T_i = \gamma^S_i$ for all $i \in X$. Note that $\Cube_{T\setminus X}(r_T)$ and $\Cube_{S\setminus X}(r_S)$ are the same point in the message space. At the same time $|T\setminus X| = p$ and $|S\setminus X| \geq p+1$, which contradicts perfect correctness.
    \end{proof}

    Using Lemma~\ref{lem:match-generalized}, this argument readily generalizes to $\EMAJ$ and $\MOD_q$ functions. This gives alternative proofs of lower bounds for these functions (in the case of $\MOD_q$ with restricted range of parameters).

    \begin{theorem}\label{thm:mainResultFromWeakerAssumptions-emaj}
        Any $(n,K)$-protocol $(\mathsf{M}, \mathsf{Rec})$ that satisfies perfect correctness and support privacy for $\EMAJ$ must have $K = \Omega(\sqrt{n})$.
    \end{theorem}

    \begin{theorem}\label{thm:mainResultFromWeakerAssumptions-mod}
        Any $(n,K)$-protocol $(\mathsf{M}, \mathsf{Rec})$ that satisfies perfect correctness and support privacy for $\MOD_q$ for $2\leq q \leq (1-c)n$ for some constant $c>0$ must have $K = \Omega(\sqrt{q})$.
    \end{theorem}    
    
    We claim that the proof for $\MAJ$ generalizes to symmetric functions $f$ satisfying the following condition: there exist $1\leq p$ and $\ell=\Omega(\sqrt p)$ with $p+\ell\leq n$ such that $f(x) \neq f(y)$ when $\|x\|=p$ and $\|y\|=p+i$ for any $i \in [\ell]$. Use the same root colorings on $S_i=[p+i]$, $i\in[\ell]$, but make the reference transcript $v \coloneq \Cube_{S_1}(r^*)$ instead of $\Cube_{[n]}(r^*)$. An agreement except at $p$ coordinates gives coincident corners at weights $p$ and $p+|S\setminus T|$ where $1 \leq |S \setminus T| \leq \ell$, as in the majority proof. For $\EMAJ_n$, take $p=\lfloor n/2\rfloor$ and $\ell=n-p$; for $\MOD_{q,n}$, take $p=q$ and $\ell=\min\{q-1,n-q\}=\Omega_c(q)$. Bounded parameters follow from the function being non-constant.
    
Observe that by Remark~\ref{rem:equivalent-functions} and~\ref{rem:conditions-for-lifting-lemma}, the hypothesis of Lemma~\ref{lem:lifting} is satisfied. Using Lemma~\ref{lem:lifting}, we lift the lower bounds on the individual message complexities $C_{\text{indiv}}(\MAJ)$, $C_{\text{indiv}}(\EMAJ)$, and $C_{\text{indiv}}(\MOD_q)$ given by Theorems~\ref{thm:mainResultFromWeakerAssumptions-maj}, \ref{thm:mainResultFromWeakerAssumptions-emaj}, and \ref{thm:mainResultFromWeakerAssumptions-mod} to obtain the following lower bounds on the PSM communication complexity:
\begin{corollary}
\label{cor:maj-emaj-modq-lower-bound}
    The functions $\MAJ_n$ and $\EMAJ_n$ have communication complexity $\Omega(n\log{n})$. The function $\MOD_q$ has communication complexity $C(\MOD_q)=\Omega(n\log q)$ when $2 \leq q\leq(1-c)n$ for any fixed constant $c>0$.
\end{corollary}
    For $\MAJ$ and $\EMAJ$, use $s=\lceil n/2\rceil$ and an integer $B=\Theta(\log n)$ below the alphabet lower bound minus one in Lemma~\ref{lem:lifting}. For $\MOD_q$ with growing $q$, use $s=\lceil(1-c/2)n\rceil$ and a sufficiently small $B=\Theta(\log q)$: every zero restriction to $m\geq s$ variables satisfies $2 \leq q\leq(1-c')m$ for $c'=c/(2-c)>0$, and $n-s+1=\Omega_c(n)$. Bounded $q\geq2$ follows from the $\Omega(n)$ correctness bound, since at fixed randomness every relevant party can choose a nonempty message.

    \subsection{General Lower Bound}
    In this section we provide the general form of our lower bound approach.


    \begin{definition}[Consistent Coloring Family]
        Consider a function $f\colon \zo^n \to \zo$, a positive integer $K$, and a bit $b\in\zo$. For any $x \in \zo^n$ consider a coloring $\alpha_x$ of $[n]$ in colors $[K]$. We say that the family of colorings $\{\alpha_x \mid x \in \zo^n\}$ is a \emph{$b$-consistent $K$-coloring family} for $f$ if the following condition holds.
        
        For any two inputs $x,y \in \zo^n$ such that $f(x)=f(y)=b$ consider the agreement set $A_{xy} = \{i \in [n] \mid \alpha_{x}(i) = \alpha_y(i)\}$. Then we require that for any $z$ such that $z_i=1$ implies $i \in A_{xy}$ we have $f(x \oplus z) = f(y \oplus z)$.
    \end{definition}

    \begin{theorem} \label{thm:general}
        Consider a function $f\colon \zo^n \to \zo$. Suppose that for some $K$ and $b\in\zo$, there is no $b$-consistent $K$-coloring family for $f$. Then there is no $(n,K)$-protocol for $f$ satisfying perfect correctness and support privacy.
    \end{theorem}

    \begin{proof}
        Suppose there is an $(n,K)$-protocol for $f$ satisfying perfect correctness and support privacy.

        Consider the message space $\cM_1\times \dots\times \cM_n$ with the hypercube structure defined above. Relabel each $\cM_i$ injectively in $[K]$. The set of inputs with output $b$ is nonempty, since otherwise every coloring family is $b$-consistent.

        Consider $b \in \{0,1\}$ from the statement of the theorem and consider a $b$-vertex $v$ in the message space, that is $v \in Z_b$. By support privacy, for any $S$ with $f(\chi_S)=b$ there is a cube $\Cube(r_S)$ such that $\Cube_{S}(r_S)=v$. Define a coloring $\alpha_{\chi_S}(i) = \Cube_{\overline{S}}(r_S)_i$, and extend the family arbitrarily to inputs with $f(x)\neq b$.
        
        By the hypothesis in the theorem this coloring cannot be $b$-consistent. That is, there are inputs $x, y \in \{0,1\}^n$ with $f(x)=f(y)=b$ and there is $S\subseteq A_{xy}$ such that $f(x \oplus \chi_S) \neq f(y \oplus \chi_S)$. Consider sets $S_x$ and $S_y$ such that $\chi_{S_x}=x$ and $\chi_{S_y}=y$. In cubes $\Cube(r_{S_x})$ and $\Cube(r_{S_y})$ we have
        \[
        \Cube_{S_x}(r_{S_x}) = \Cube_{S_y}(r_{S_y}) = v.
        \]
        Then by the coloring agreement we also have
        \[
        \Cube_{S_x \triangle S}(r_{S_x}) = \Cube_{S_y  \triangle S}(r_{S_y}),
        \]
        which is a contradiction with perfect correctness since
        \[
        f(\chi_{S_x \triangle S}) = f(x \oplus \chi_S) \neq f(y \oplus \chi_S) = f(\chi_{S_y \triangle S}).
        \]
        
        \end{proof}

    We next show that the technique outlined in Sections~\ref{sec:agreement} and \ref{sec:lower-bounds} subsumes the Nechiporuk bound of~\cite{ball2020complexity}.

    \begin{theorem}[Nechiporuk Bound] \label{thm:nechiporuk}
        Consider a function $f\colon \zo^n \to \zo$ and $S \subseteq [n]$. Denote by $N_S(f)$ the number of subfunctions of $f$ on $S$. Then every perfectly correct, support-private protocol for $f$ with local alphabet sizes $K_i$ satisfies $\sum_{i\in S}\log K_i\geq\log N_S(f)-1$.
    \end{theorem}

    Before proving Theorem~\ref{thm:nechiporuk} we observe that the following corollary is immediate.
    \begin{corollary}
        Consider a function $f\colon \zo^n \to \zo$ and let $V = (V_1, \ldots, V_m)$ be a partition of $[n]$, that is $[n] = \bigsqcup_{i=1}^{m} V_i$. Then 
        $$
        C_\text{indiv-sum}(f) \geq \sum_{i=1}^m \log N_{V_i}(f) - m.
        $$
    \end{corollary}

    To address the communication metric in Theorem~\ref{thm:nechiporuk}, we analyze the communication cubes directly rather than apply Theorem~\ref{thm:general}.

    \begin{proof}[Proof of Theorem~\ref{thm:nechiporuk}]
        Denote the size of $S$ by $k$ and for notational convenience assume that $S = [k]$.
    
        Without loss of generality assume that at least half of the subfunctions of $f$ on $S$ output $0$ on $0$-input. Denote the set of these subfunctions by $\mathcal F_S$. In particular, $|\mathcal F_S| \geq N_S(f)/2$. For each function $g \in \mathcal F_S$ fix an input $z \in \zo^{n-k}$ on $\overline S$ such that $g = f \mid_{\overline S = z}$. Denote the set of all of these inputs by $Z \subseteq \zo^{n-k}$. In particular, $|Z| = |\mathcal F_S| \geq N_S(f)/2$.

        In the communication cube consider an arbitrary supported $0$-transcript $v\in Z_0$. By support privacy for any $0$-input $x \in \zo^n$ of $f$ there is a cube that has $v$ as an $x$-corner. For each $z \in Z$ consider input $(\vec 0, z)$ to $f$ (note that it is a $0$-input) and fix such a cube $C_z$.
    
        Now for each $z \in Z$ consider the input $(\vec 1, z)$.
        For the sake of contradiction suppose $\prod_{i\in S}K_i<N_S(f)/2$. Then there are fewer than $N_S(f)/2$ possible tuples of messages of these players. Thus, for two inputs in $Z$ these transcripts are the same. Denote these inputs by $z_1$ and $z_2$ and their corresponding functions in $\mathcal F_S$ by $g_1$ and $g_2$.

        There is an input $y \in \zo^k$ such that $g_1(y) \neq g_2(y)$. As a result,
        $$
        f(y,z_1) = g_1(y) \neq g_2(y) = f(y,z_2).
        $$
        At the same time, by the choice of $z_1$ and $z_2$ the corner $(y,z_1)$ in $C_{z_1}$ and the corner $(y,z_2)$ in $C_{z_2}$ are the same vertex in the communication cube. This is a contradiction with the perfect correctness.
    \end{proof}

    \section{Upper Bounds}
    \label{sec:upper-bounds}

    \subsection{\texorpdfstring{Perfect Privacy for any Symmetric Function in $\widetilde{O}(n^{3/2})$ bits}{Perfect Privacy for any Symmetric Function}}
    \label{subsec:upper-bound:perf-priv-sym}

    In this section, we prove the following theorem.

    \begin{theorem}
    \label{thm:upper-bound-perf-private-sym}
        For any symmetric Boolean function $f \colon \{0,1\}^n \to \{0,1\}$ there is a perfectly correct, perfectly private PSM of total complexity $\widetilde{O}(n^{3/2})$.
    \end{theorem}

    \begin{proof}
        We compose an outer protocol, which evaluates the function from residues of the Hamming weight, with inner protocols that privately compute the outer messages.
        This follows Eriguchi~\cite{Eriguchi2026}, who composes an outer PSM on histogram coordinates with a private selector for each outer message. We use small CRT residues of the Hamming weight as the outer inputs, which reduces the size of the selector tables.

        Write $f(x)=g(\|x\|)$ for $g:\{0,\ldots,n\}\to\{0,1\}$.

        The outer protocol uses Theorem~\ref{thm:bkn-psm} to evaluate $g$ on the weight reconstructed from its residues modulo small primes. We assume $n$ is sufficiently large. (For the finitely many smaller $n$, we can use any universal perfectly private PSM.)

        Let $p_1, \ldots, p_k$ be distinct prime numbers and let $Q = \max_i p_i$. We want the following properties of this set of primes:
        \begin{enumerate}
            \item $p_1 \ldots p_k > n$,
            \item $Q^k \leq n \polylog(n)$.
        \end{enumerate}

        It is not hard to see that these properties are achievable with $k, Q = \polylog(n)$ (see Lemma~\ref{lem:small-primes} in Appendix~\ref{subsec:prime-choice}).

        Using Chinese Remainder Theorem we can define a function
        \[
        \CRT\colon \ZZ_{p_1} \times \ldots \times \ZZ_{p_k} \to \{0,1, \ldots, n\},
        \]
        such that $z=\CRT(y_1,\ldots,y_k)$ whenever $z\in\{0,\ldots,n\}$ satisfies $z=y_i\pmod{p_i}$ for all $i$; set the value to $0$ if no such $z$ exists.

        The outer protocol will be computing the function
        \[
        h = g \circ \CRT \colon \ZZ_{p_1} \times \ldots \times \ZZ_{p_k} \to \{0,1\}.
        \]
        By Theorem~\ref{thm:bkn-psm}, extending this function arbitrarily to $[Q]^k$, there is a protocol for this function, in which $k$ players with inputs $y_1, \ldots, y_k$ use shared randomness $R$ and send messages $m_1(y_1,R), \ldots, m_k(y_k,R)$ to the referee, who is then able to compute $h(y_1, \ldots, y_k)$ without learning any other information about $y_1, \ldots, y_k$. Encode each outer message at a fixed length by rounding up the logarithm of its alphabet size. If we denote the length of the message $m_i$ by $l_i$, then for this protocol we have
        \[
        l_1 + \ldots + l_k = O(k^3 Q^{k/2}) = \widetilde{O}(\sqrt{n}).
        \]

        However, in our setting we have $n$ players that have inputs $x_1, \ldots, x_n \in \{0,1\}$. The obstacle is that no individual player knows the residues $y_i = \sum_j x_j \pmod{p_i}$.

        Each inner protocol lets the referee recover one outer message $m_i(y_i,R)$ without revealing the residue $y_i$. Each invocation uses fresh inner randomness, independent of $R$ and the other inner invocations. The referee then applies the outer decoder to these messages.

        For the inner protocol we arrive at the following more general problem. We have $n$ players holding private bits $x_1, \ldots, x_n$, and also sharing a secret function $T \colon \ZZ_p \to \{0,1\}^l$ (in the outer protocol, $T=m_i(\cdot,R)$, where $R$ is unknown to the referee). The goal of the players is to reveal $T\left(\sum_i x_i \pmod{p}\right)$ to the referee without revealing any other information.

        We now describe the inner protocol. The players sample independent uniform elements $s_1,\ldots,s_n\in\ZZ_p$ to mask the input bits. Denote $s = \sum_i s_i \pmod{p}$. The players also generate the family of uniformly random functions
        \[
        R_i \colon \ZZ_p \to \{0,1\}^l
        \]
        for $i = 1, \ldots, n-1$.

        The first player sends $a_1 = x_1 + s_1$ and $R_1(a_1)$.

        The $i$-th player for $i=2, \ldots, n-1$ sends $a_i = x_i + s_i$ as well as the whole table of the function $Q_i(y) = R_{i-1}(y - a_i) \oplus R_{i}(y)$.

        The last player sends $a_n = x_n + s_n$ and the whole table of the function $Q_n(y) = R_{n-1}(y - a_n) \oplus T(y - s)$.

        Each message depends only on the sender's input bit, the shared function $T$, and the shared randomness, so the parties can compute their messages simultaneously.

        The referee computes
        \[
        R_1(a_1) \oplus Q_2(a_1 + a_2) \oplus Q_3(a_1 + a_2 + a_3) \oplus \ldots \oplus Q_n(a_1 + \ldots + a_n),
        \]
        which after expansion simplifies to
        \[
        T(a_1 + \ldots + a_n - s) = T(x_1 + \ldots + x_n).
        \]

        To verify perfect privacy, fix the input bits and the table $T$. Conditional on the masked inputs $a_1,\ldots,a_n$, the map from $(R_1,\ldots,R_{n-1})$ to $(Q_2,\ldots,Q_n)$ is a bijection: recover $R_{n-1}$ from $Q_n$ and $T$, then recover the remaining $R_i$ backwards. Hence the $Q_i$ are independent uniform tables. Given only the output $t=T(\sum_i x_i)$, a simulator samples the masked inputs and tables uniformly and sets the first party's remaining message to
        \[
            b=t\oplus\bigoplus_{i=2}^n Q_i(a_1+\cdots+a_i).
        \]
        The telescoping identity above gives $b=R_1(a_1)$ in the real protocol, so this simulation has exactly the real transcript distribution. Composing it with the outer simulator proves perfect privacy of the full protocol.

        For communication, sending all masked inputs $a_i$ requires at most $n\lceil\log p\rceil$ bits of communication and sending the function tables requires $npl$ bits of communication. In total, to implement the protocol for one function $T$ they send at most
        $n\lceil\log p\rceil+npl$ bits.

        Applying the inner protocol to all messages $m_1,\ldots,m_k$ of the outer protocol gives total communication at most
        $n k\lceil\log Q\rceil+nQ\sum_i l_i=\widetilde{O}(n^{3/2})$ bits.
    \end{proof}

    \subsection{Support-Private Protocols for Symmetric Functions}
    \label{subsec:upper-bound-sup-priv-symm}
    In this section, we prove the following theorem.

    \begin{theorem}
    \label{thm:support-priv-upper-bound}
        For any symmetric Boolean function $f \colon \{0,1\}^n \to \{0,1\}$ there is a perfectly correct, support-private PSM of total complexity $\widetilde{O}(n)$.
    \end{theorem}

    We use the same composition as in Section~\ref{subsec:upper-bound:perf-priv-sym}, replacing the outer protocol with the support-private protocol of Theorem~\ref{thm:weakly-private-psm}.

    \begin{proof}
        As in the perfectly private construction, it suffices to handle large enough $n$. For these large enough $n$, write $f(x)=g(\|x\|)$ and choose the distinct primes from Lemma~\ref{lem:small-primes}, so that the following properties hold.
        \begin{enumerate}
            \item $p_1 \ldots p_k > n$,
            \item $Q^k \leq n \polylog(n)$
        \end{enumerate}
        where $Q = \max_i{p_i}$.
      As in that construction, we use the Chinese Remainder Theorem to define a function
        \[
        \CRT\colon \ZZ_{p_1} \times \ldots \times \ZZ_{p_k} \to \{0,1, \ldots, n\},
        \]
        that returns the unique representative in $\{0,\ldots,n\}$ when one exists, and $0$ otherwise. Let $\CRT'$ denote a version of this function with binarized inputs, assigning $0$ to invalid residue encodings. In other words, define
        \[
        \CRT' \colon \{0, 1\}^{\lceil \log p_1\rceil} \times \ldots \times \{0, 1\}^{\lceil \log p_k\rceil} \to \{0,1,\ldots,n\}.
        \]

        The outer protocol will be computing the function

         \[
        h = g \circ \CRT' \colon \{0, 1\}^{\lceil \log p_1\rceil} \times \ldots \times \{0, 1\}^{\lceil \log p_k\rceil} \to \{0,1\}.
        \]

        Now since $Q^k \leq n\polylog(n)$, we have that $L=\sum_{i=1}^k\lceil\log p_i\rceil=O(\log n)$. By Theorem~\ref{thm:weakly-private-psm}, there is a protocol for $h$ with communication $\widetilde{O}(L^2) = \widetilde{O}(\log^2{n})$. The protocol provided by Theorem~\ref{thm:weakly-private-psm} will contain $L$ messages, since it assumes each party holds only one bit. First give each virtual message a fixed-length encoding, with a length prefix if needed; the possible factor $L$ is polylogarithmic. However, we will concatenate components associated with the same prime residue to obtain $k$ messages $m_1(y_1, R), \ldots, m_k(y_k, R)$, where $m_i(y_i, R)$ is a concatenation of the messages for all the players whose bit encodes the value of the residue $y_i$.

    As in the proof of Theorem~\ref{thm:upper-bound-perf-private-sym}, we cannot just use the protocol above because the players do not individually know the inputs $y_i = \sum_j x_j \pmod{p_i}$. Instead of implementing the outer protocol above explicitly, our players will use an inner protocol to compute messages $m_i$ for all $i$. More specifically, they will send the information to the referee to compute each message $m_i$ (by inner protocol), and after that the referee will be able to use the outer protocol to compute $h(y_1,\ldots, y_k)$. Letting $\ell_i$ denote the length of $m_i(y_i,R)$, we have from the upper bound on the outer protocol that $\sum_i \ell_i = \polylog(n)$. The inner protocol from Theorem~\ref{thm:upper-bound-perf-private-sym} suffices to send the information for all messages $m_1, \ldots, m_k$, so we omit the details here. Use independent inner randomness for these invocations. For an outer transcript $m$, let $K(m)$ denote the distribution of its inner simulations. The full transcript support is the union of $\operatorname{Supp}K(m)$ over possible outer transcripts $m$; equal outer supports therefore imply equal full supports. This gives us an upper bound of $n k\lceil\log Q\rceil+nQ\sum_i\ell_i=\widetilde{O}(n)$ bits of communication.
    \end{proof}
    
    \subsection{Upper Bound for Majority}
    \label{subsec:upper-bound-branching-programs}

    Combining a known permutation branching program for $\MAJ_n$ with a standard PSM construction gives communication $O(n\log^3 n\log\log n)$.

    \begin{definition}[strict permutation branching program] A \emph{strict permutation branching program} of width $w$ and length $\ell$ consists of the following: 
    \begin{enumerate} 
    \item A sequence of variable indices $i_1, \ldots, i_\ell$. 
    \item For each layer $j$ and each bit $b$, a permutation $P_j(b)$ of $\{1, \ldots, w\}$.
    \item A start state $s_0$.
    \item An accepting set $A \subseteq \{1, \ldots, w\}$. 
    \end{enumerate}
    On input $x$ it outputs 1 if and only if the image of $s_0$ under permutation $P_\ell(x_{i_\ell}) \ldots P_1(x_{i_1})$ is in $A$. 
    \end{definition}

    The following lemma is implicit in Klauck and Podder~\cite{KlauckP14-conf} (the proof can be found in Appendix 10.5 of the full version of the paper~\cite{KlauckP14-full}).
    \begin{lemma}[Klauck and Podder \cite{KlauckP14-conf}]
        The function $\MAJ_n$ can be computed by a strict permutation branching program of width $O(\log n)$ and length $O(n \log^2 n)$.
    \end{lemma}

    The next lemma summarizes the standard conversion of permutation branching programs into perfectly private PSM protocols. The construction goes back to Kilian~\cite{Kilian1988}; see also~\cite[Section 4.1.3]{ishai2013randomization} for its formulation over general groups. Papakonstantinou and Yang~\cite{PapakonstantinouY14} give the permutation version for width $5$; the same construction generalizes to arbitrary width. We give the details in Appendix~\ref{subsec:pbp-proof}.

    \begin{lemma}[Kilian \cite{Kilian1988}]\label{lem:pbp-psm}
        If $f$ is computable by a strict permutation branching program of width $w$ and length $\ell$, with a fixed initial state and a fixed final accepting set, then there is a perfectly private simultaneous messages protocol for $f$ with total communication $O(w\ell \log w)$.
    \end{lemma}

    These two facts readily give us the following corollary.

    \begin{corollary}
    \label{cor:maj-upper-bound}
        There is a perfectly private simultaneous messages protocol for $\MAJ_n$ with total communication $O(n \log^3 n \log\log n)$.
    \end{corollary}

\bibliographystyle{alpha} 
\bibliography{mybib}

    \appendix
    \section{Derivation of Combinatorial Results}\label{sec:combinatorial}
    \newcommand{\subseq}{\preceq}
    \subsection{Simple Higman's Lemma Proof}
    For strings $u,v$, write $u\subseq v$ when $u$ is a subsequence of $v$.

    We first recall Higman's Lemma (for finite alphabets and ordered by equality).
    \begin{lemma}[Higman's Lemma~\cite{higman,beigel1990}]\label{lem:higman}
           If $\Sigma$ is a finite alphabet, then every infinite sequence of finite strings
        \[w_1,w_2,w_3,\ldots\]
        contains $i<j$ such that $w_i\subseq w_j$.   
    \end{lemma}
    We use K\"onig's lemma to derive a finite bound from Higman's infinite-sequence theorem.
    \begin{lemma}[K\"onig's Lemma~\cite{konig1927schlussweise}; see also Lemma 8.1.2 in~\cite{diestel2017graph}]\label{lem:konig}
        Let $\cT$ be an infinite rooted tree in which every vertex has finitely many children. Then $\cT$ contains an infinite path starting at the root.
    \end{lemma}
    The finite variant follows by applying K\"onig's lemma to the tree of finite bad sequences.
    \begin{proof}[Proof of (Simple) Higman's Lemma, Lemma~\ref{lem:Higman+}]
        Call a sequence $w_1,\ldots,w_\ell$ such that $|w_i|=i$ ``bad'' if for all $i<j$, $w_i \not\subseq w_j$.

        Suppose the lemma is false, so a bad sequence of length $\ell$ exists for every $\ell$.

        Note that if a sequence is ``bad,'' any prefix of that sequence is also ``bad.''

        Thus we can build an arbitrarily deep tree of ``bad'' sequences by letting vertices at depth $\ell$ be ``bad'' sequences of length $\ell$, $w_1,\ldots,w_\ell$, with parent being the $(\ell-1)$-length prefix, $w_1,\ldots,w_{\ell-1}$ (which is also ``bad''). Note that this tree is finitely branching because a node at depth $\ell$ can only have $|\Sigma^{\ell+1}|$ many children (because we required $|w_i|=i$).
        Hence by Lemma~\ref{lem:konig}, $\cT$ contains an infinite path from which we can derive an infinite ``bad'' sequence, a sequence of strings
        \[w_1,w_2,w_3,\ldots \qquad \mbox{(with $|w_i|=i,\forall i)$}\]
        such that there does not exist $i<j$ such that $w_i\subseq w_j$, contradicting Lemma~\ref{lem:higman}.
    \end{proof}

    \subsection{Iterated Finite Ramsey}
        We first recall the finite Ramsey theorem.
    \begin{theorem}[Finite Ramsey's Theorem~\cite{ramsey,Graham2015Ramsey}]\label{thm:ramsey}
        For all positive integers $k,m,r$ there is an integer $R_k(m;r)$ such that whenever the $k$-subsets of a set $X$ of size at least $R_k(m;r)$ are colored with $r$ colors, then there is a subset
        \[ Y\subseteq X,\qquad |Y|=m,\]
        such that all $k$-element subsets of $Y$ have the same color.
    \end{theorem}
    The desired form follows by induction.
    \begin{proof}[Proof of Theorem~\ref{thm:Ramsey+}]
        Let $r_k=|\Sigma|^k$ be the number of possible colors for $k$-sets. Define
        \begin{align*}
            N_m&:=m,
            N_{k-1} := R_k(N_k;r_k) \mbox{ for $k=m,m-1,\ldots,1$}.
        \end{align*}
        Take $N(\Sigma,m) = N_0$.
    Now, we iteratively invoke Theorem~\ref{thm:ramsey}. Start with $X_0$ a set of size $N_0$. Then inductively, at step $k$ we have a set $X_{k-1}$ of size $N_{k-1}=R_k(N_k;r_k)$ so by Theorem~\ref{thm:ramsey} we have that there must be some
    \[
        X_k \subseteq X_{k-1}, \qquad |X_k|=N_k \mbox{ where all $k$-sets have the same color $u_k\in\Sigma^k$}.
    \]
    
    Moreover, by the inductive hypothesis, for $\ell$ such that $1\le \ell \le k-1$, $\ell$-subsets of $X_k$ have the same color because $X_k$ is a subset of $X_{k-1}$. Hence $\ell$-subsets of $X_k$ are the same color for any $1\le \ell\le k$.
    
    Thus, we have constructed a tower of sets
        \[
            X_0\supseteq X_1\supseteq\cdots\supseteq X_m, \qquad |X_k| = N_k,
        \]
        where $H=X_m$ has size $N_m = m$ and for any $k\le m$, all $k$-subsets have the same color.
    \end{proof}

    \subsection{Choice of Small Primes}\label{subsec:prime-choice}
    \begin{lemma}\label{lem:small-primes}
        For all sufficiently large $n$, there are distinct primes $p_1,\ldots,p_k$, with $Q=\max_i p_i$, such that $p_1\cdots p_k>n$, $Q^k=O(n\log^3 n)$, and $k,Q=\polylog(n)$.
    \end{lemma}
    \begin{proof}
        Put $t=\lceil\log n\rceil$ and $X=t^3$. Standard prime-counting estimates, for example~\cite[Theorem 6.9]{Dusart2010}, give $\pi(2X)-\pi(X)=\Omega(X/\log X)$. Partition $(X,2X]$ into $t$ intervals of length $X/t$. One interval contains $\Omega(X/(t\log X))$ primes, more than $\lfloor\log n/\log X\rfloor+1$ for sufficiently large $n$. Select primes in this interval until their product $P$ first exceeds $n$. Then
        \[
            k\leq\lfloor\log n/\log X\rfloor+1,\qquad Q\leq2X,\qquad P\leq nQ.
        \]
        Since all selected primes lie in one such interval, $Q/p_i\leq1+1/t$, and hence
        \[
            Q^k=P\prod_i(Q/p_i)\leq nQ(1+1/t)^k=O(nX).
        \]
    \end{proof}

    \subsection{Permutation Branching Program Conversion}\label{subsec:pbp-proof}
    \begin{proof}[Proof of Lemma~\ref{lem:pbp-psm}]
        Write the transitions as permutations $P_j(x_{i_j})$ of $[w]$, with initial state $s_0$ and accepting set $A$. Sample independent uniform permutations $\rho_0,\ldots,\rho_\ell$. The owner of $x_{i_j}$ sends $E_j=\rho_jP_j(x_{i_j})\rho_{j-1}^{-1}$; one party also sends $\alpha=\rho_0(s_0)$ and $A'=\rho_\ell(A)$. The evaluator tests whether $E_\ell \cdots E_1(\alpha)\in A'$, giving perfect correctness.

        Put $P=P_\ell\cdots P_1$ and $E=E_\ell\cdots E_1$. The map from the masks to $(\rho_0,E_1,\ldots,E_\ell)$ is bijective, so the $E_j$ and $\rho_0$ are independent uniform permutations. Conditional on the $E_j$, the pair
        \[
            (\alpha,E^{-1}A')=(\rho_0(s_0),\rho_0(P^{-1}A))
        \]
        is a uniformly relabeled point and set, whose law depends only on whether $P(s_0)\in A$. A simulator therefore samples the $E_j$ and $\alpha$ uniformly, then a uniform $|A|$-element set $A'$ containing $E(\alpha)$ exactly when the output is one. Communication is
        \[
            \ell\lceil\log(w!)\rceil+\lceil\log w\rceil+w=O(w\ell\log w).
        \]
    \end{proof}

    \paragraph{AI usage disclosure.} Work on this project began in spring 2023. The authors made initial progress on the complexity of $\OR_n$ for small constant $n$ and suspected that extremal combinatorics might provide the key to stronger bounds. On April 29, 2026, ChatGPT-5.5 Pro was used to find a viable superlinear lower bound for the $\OR$ function (Theorem~\ref{thm:mainResultFromWeakerAssumptions}), indeed using classic results by Ramsey and Higman. The authors subsequently generalized and extended the lower bound approach without AI. Upper bounds were obtained with partial AI assistance, exploring branching program results (Section~\ref{subsec:upper-bound-branching-programs}), and parameter verification in composition constructions (i.e., Section~\ref{subsec:upper-bound:perf-priv-sym} and Section~\ref{subsec:upper-bound-sup-priv-symm}). Finally, ChatGPT-6.2 Astra was used in November 2026 to obtain the tight lower bound for the $\OR$ function, presented in Section~\ref{sec:new-lower-bounds}.
    
    AI tools were used for proof verification and editorial assistance. The authors take full responsibility for the contents. 

\end{document}